\documentclass{article}

\usepackage[a4paper,
            bindingoffset=0.2in,
            left=1in,
            right=1in,
            top=1in,
            bottom=1in,
            footskip=.25in]{geometry}

\usepackage{url} 
\usepackage{hyperref}

\usepackage[T1]{fontenc}
\usepackage[utf8]{inputenc} 

\usepackage{amsmath}
\usepackage{amssymb}
\usepackage{amsthm}
\usepackage{mathtools}

\usepackage[capitalise, noabbrev]{cleveref}

\usepackage{natbib}
\usepackage{graphicx}

\usepackage{float}
\usepackage{bm}
\usepackage{accents} 
\usepackage{pdfpages}
\usepackage{esint}
\usepackage{multirow}
\usepackage{xcolor}
\usepackage{enumitem}
\usepackage{caption}
\usepackage{pdfsync}
\usepackage{subcaption}

\usepackage{stmaryrd}

\usepackage[colorinlistoftodos]{todonotes}
\usepackage{dsfont}
\usepackage{algpseudocode}

\theoremstyle{plain}
\newtheorem{thm}{Theorem}
\newtheorem{prop}{Proposition}
\newtheorem{corollary}{Corollary}
\newtheorem{lemma}{Lemma}

\theoremstyle{definition}

\newtheorem{definition}{Definition}
\newtheorem{remark}{Remark}
\newtheorem{example}{Example}
\newtheorem*{properties}{Properties}

\floatstyle{ruled}
\newfloat{algorithm}{tbp}{loa}
\providecommand{\algorithmname}{Algorithm}
\floatname{algorithm}{\protect\algorithmname}

\newcommand{\R}{\mathbb{R}}

\newcommand{\E}{\mathbb{E}}
\newcommand{\var}{\mathrm{Var}}

\newcommand{\ind}{\mathds{1}}
 \newcommand{\smallo}{{\scriptscriptstyle\mathcal{O}}} 

\DeclareMathOperator*{\argmin}{argmin}

\newcommand{\crl}[1]{\left\{#1 \right\}}
\newcommand{\round}[1]{\left (#1 \right)}
\newcommand{\sqrd}[1]{\left[#1 \right]}

\newcommand\inner[2]{\langle #1, #2 \rangle}

\newcommand{\Liminf}{\mathrm{Liminf}}

\begin{document}
	
    \title{Gradient-free optimisation via integration}

\author{\thanks{c.andrieu@bristol.ac.uk} Christophe Andrieu, \thanks{nicolas.chopin@ensae.fr} Nicolas Chopin, \thanks{ettore.fincato@bristol.ac.uk} Ettore Fincato, \thanks{mathieu.gerber@bristol.ac.uk} Mathieu Gerber\\
        \small $^{*,,\ddagger,\S}$ School of Mathematics, University of Bristol, UK\\
        \small $^{\dagger}$ ENSAE, Institut Polytechnique de Paris, France\\
}

	
\maketitle
	\begin{abstract}

		We develop and analyse an approach to optimize functions $l\colon \mathbb{R}^d
			\rightarrow \mathbb{R}$ not assumed to be convex,  differentiable or even
		continuous. The algorithm belongs to the class of model-based search methods. The idea is to fit recursively $l$ to a parametric family of distributions,  using a Bayesian update followed by a reprojection back
		onto the chosen family. Remarkably, reprojection in our scenario boils down to computing
		expectations, which can be simply approximated through Monte Carlo. We
		show that when the family of distributions is appropriately chosen this approach can be interpreted as an implicit time-inhomogeneous
		gradient descent algorithm on a sequence of smoothed approximations of $l$, providing a  route to establishing convergence. We
		establish new results for generic inhomogeneous gradient descent algorithms, which we specialise to the model-based search algorithm in the Gaussian scenario. We illustrate the performance of the algorithm on a
		challenging classification task in machine learning.

	\end{abstract}

	\begin{keywords}  
Gradient-Free Optimisation, Bayesian Updating, Variational methods, Smoothing
	\end{keywords}


\vspace{0.5cm}

All notation can be found in \cref{sec:notation}.

\section{Introduction} \label{sec:introduction}

\subsection{Motivation}

Let $l\colon\mathsf{X} \coloneq \mathbb{R}^{d}\rightarrow\mathbb{R}$ be a
lower semi-continuous,  potentially non-differentiable function such that
$\inf_{x\in\mathbb{R}^{d}}l(x)>-\infty$ and hence $\arg\min_{x\in
		K}l(x)\neq\emptyset$ for any compact set $K \subset \mathsf{X}$. This
paper is concerned with gradient-free algorithms to minimize such a function $l$,
provided it can be evaluated pointwise.

The algorithm we study is based on the following central idea. Let $\phi$ be the density of the
standard normal distribution $\mathcal{N}(0,\mathbf{I}_d)$, and let
\begin{equation}
	\label{eq:gaussian_family}
	\pi_{\theta,\gamma}(x)
	\coloneq \frac{1}{\gamma^{d/2}}\phi\left(\frac{1}{\sqrt{\gamma}}(x-\theta)\right)
\end{equation}
for $\theta\in\Theta\coloneq\mathbb{R}^{d}$, $\gamma>0$; that is, the density
of $\mathcal{N}(\theta, \gamma \mathbf{I}_d)$ distribution.
Then, for a sequence
$\{\gamma_{n}\geq0,n\in\mathbb{N}\}$ such that $\gamma_{n}\downarrow0$, define
sequentially the families of distributions $\{\pi_{n},n\in\mathbb{N}\}$ and
$\{\tilde{\pi}_{n},n\in\mathbb{N}\}$ as in Algorithm~\ref{alg:main_det_algo}.

\begin{algorithm}[h]

	\caption{Gradient-Free Ideal Algorithm}
	\label{alg:main_det_algo}

	\begin{algorithmic}
		\Require objective function $l$, initial parameter $\theta_0$, stepsizes
		$(\gamma_n)_{n\geq 0}$.
		\State $\pi_0 \gets \pi_{\theta_0, \gamma_0}$ \Comment{initial distribution}
		\While{$n \geq 0$}
		\State $\tilde{\pi}_{n+1}(x)\propto \pi_{n}(x)\exp\left\{-l(x)\right\}$
		\Comment{generalised Bayesian rule}
		\State $\theta_{n+1}\in \argmin_{\theta\in \Theta} \mathrm{KL}(\tilde{\pi}_{n+1}, \pi_{\theta,\gamma_{n}})$
		\Comment{project with Kullback-Leibler divergence}
		\State $\pi_{n+1} \gets\pi_{\theta_{n+1},\gamma_{n+1}}$ \Comment{distribution shrinking}

		\EndWhile
		\Ensure sequence of distributions $\tilde{\pi}_n$ and parameters $\theta_n$.
	\end{algorithmic}

\end{algorithm}

An iteration of Algorithm~\ref{alg:main_det_algo} therefore consists of the
application of Bayes' rule, where $l$ plays the role of a negative
log-likelihood and $\pi_{n}$ that of the prior distribution, followed by a
``projection'' onto the normal family $\pi_{\theta, \gamma_{n}}$, using the
Kullback-Leibler divergence as a criterion.
As illustrated in Figure~\ref{fig:varying_gaussians}, Bayes' rule tilts
$\pi_n=\mathcal{N}(\theta_n,\gamma_n\mathbf{I}_d)$ towards regions where $l$ is
small; the resulting ``posterior'' distribution $\tilde{\pi}_{n+1}$ is then
approximated by a Gaussian $\mathcal{N}(\theta_{n+1},\gamma_n\mathbf{I}_d)$.
Combined with the reduction of variance,
$\pi_{n+1}=\mathcal{N}(\theta_{n+1},\gamma_{n+1}\mathbf{I}_d)$,  
the sequence of means $\{\theta_{n},n\in\mathbb{N}\}$ is 
expected to converge to a
local minimum. (This algorithm can be extended to the scenario where unbiased noisy measurements of $l$ are available, see  Section~\ref{sub:accounting_noise}.)

\begin{figure}[H]
	\centering
	\includegraphics[scale=0.30]{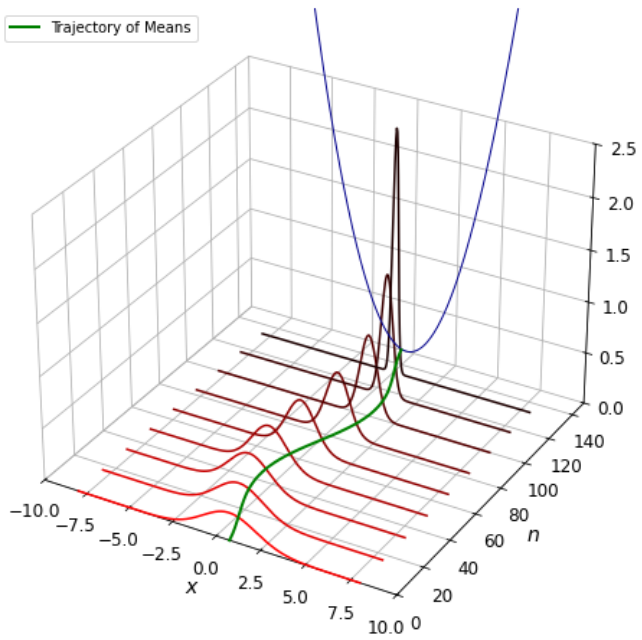}
	\caption{Illustration of the  algorithm in the Gaussian case, when
		$l(x)=x^2$ (blue line in the background). The light red curve is the
		initial Gaussian density ($n=0$). The red to black curves are the Gaussian
		densities $\pi_{n}$, which are recursively obtained by KL minimisation.
		(The mean of these Gaussian densities is plotted in green.)
		We observe that these Gaussian densities are progressively attracted to low
		values of $l$.
		\label{fig:varying_gaussians}}
\end{figure}

In practice, one may approximate Algorithm~\ref{alg:main_det_algo} with a cloud of $N\in\mathbb{N}$ random samples $\{X_{n}^{i},i\in\llbracket N\rrbracket\}$
propagated along the iterations; 
see \cref{sec:experiments} for details.
We focus on establishing convergence of the ideal Algorithm~\ref{alg:main_det_algo}, corresponding to the scenario
$N\rightarrow\infty$.
We see the study of such ideal algorithms
as a prerequisite to the study of their
implementable versions, seen as perturbations of the ideal algorithms.

\subsection{Connection with gradient descent}

In a standard statistical context, repeated application of Bayes' rule is
known to lead to a concentration phenomenon around particular maximum points or
the posteriors, under general conditions \citep{vanderVaart2012}.

The update considered here differs from standard Bayesian updating in that it
involves a reprojection step, therefore necessitating a new approach to
establishing ability of the algorithm to find minima of $l$. This
reprojection step is motivated by practical considerations: it
circumvents the need to propagate the sequence of distributions obtained by
repeated use of Bayes' update and greatly facilitates implementation (as we
elaborate below). Notably, the reprojection step also offers an entry point for the theoretical analysis of the algorithm. The
crucial observation allowing us to prove convergence of
Algorithm~\ref{alg:main_det_algo}, in the sense that $\pi_{n}$ concentrates on local
minima of $l$, is that it implicitly implements a steepest descent
algorithm tracking the minima of a sequence of differentiable approximations
$\{l_n\colon \mathbb{R}^d \rightarrow \mathbb{R},n\geq1\}$ of $l$.
When such approximations converge to $l$, validity of the procedure
should ensue.

More precisely, the reprojection step can be shown to correspond to so-called
moment matching, a fact extensively used in variational inference
\citep{wainwrightjordan2008}. Taking into account that in the present setup
$\theta_{n+1}$ is the first order moment of $\pi_{n+1}$, or mean, moment
matching takes the form

\begin{equation}
	\label{eq:intro-moment-match}
	\theta_{n+1}
	= \int x\tilde{\pi}_{n+1}(x){\mathrm d}x
	=\theta_{n}+\frac{\int(x-\theta_{n})\exp\big\{-l(x)\big\}\phi\big(\frac{x-\theta_{n}}{\sqrt{\gamma_{n}}}\big){\mathrm d}x}{\int\exp\big\{-l(x)\big\}\phi\big(\frac{x-\theta_{n}}{\sqrt{\gamma_{n}}}\big){\mathrm d}x}\,.
\end{equation}


It is the evaluation of these expectations which in practice requires a Monte
Carlo approximation with weighted samples. For $(\theta,\gamma)\in\Theta\times\mathbb{R}_{+}$, let
\begin{equation}
	\label{eq:Laplace_gaussian}
	l_\gamma(\theta) \coloneq -\log\left(\int 
    \frac{\exp\left\{-l(x)\right\}}{\gamma^{d/2}}
    \phi\Bigl(\frac{x-\theta}{\sqrt{\gamma}}\Bigr)\mathrm{d} x\right),
\end{equation}
and for $n\in\mathbb{N}$ and $\theta\in\Theta$ let $l_n(\theta):=l_{\gamma_{n}}(\theta).$
Then one can write \eqref{eq:intro-moment-match} in the familiar
form
\begin{equation}
	\label{eq:non-homogeneous-descent-generic}
	\theta_{n+1}=\theta_{n}-\gamma_{n}\nabla l_n(\theta_{n})\,,
\end{equation}
and recognize a time inhomogeneous steepest-descent algorithm tracking the
sequence of stationary points of the sequence of functions
$\{l_n,n\in\mathbb{N}\}$, again smoothed versions of $l$. It is
remarkable that while this interpretation provides us with an additional
rational for Algorithm~\ref{alg:main_det_algo} and a route to establishing its
convergence for a large class of non-differentiable functions $l$ (the
subject of next sections), implementation does
not require differentiation but instead integration.

The arguments we provide in \cref{sec:Convergence:-deterministic-scena,sec:Laplace_functionals,sec:convergence_of_algorithm} lead
to the following convergence result on the inhomogeneous gradient descent
in~\eqref{eq:non-homogeneous-descent-generic}.


\begin{thm}
	\label{thm: general conv gaussian}
	Let $l:\R^d\to \R$ be lower-bounded, strongly lower semi-continuous and
	assume there exists $C_l\in(0,\infty)$ such that
	$|l(\theta')-l(\theta)|\leq C_l+C_l\|\theta'-\theta\|^2$ for all
	$\theta,\theta'\in\R^d$. Let $\{\gamma_n,\,n\in\mathbb{N}\}$ be a sequence
	with $\gamma_n=n^{-\beta}$ for all $n\geq 1$, for some $\beta\in (0,1)$.
	Let $\crl{\theta_n,\, n\in \mathbb{N}}$ be the output of recursion \eqref{eq:non-homogeneous-descent-generic}, corresponding to Algorithm~\ref{alg:main_det_algo}. Then, there exists a subsequence $\crl{\theta_{n_k},\, k\in \mathbb{N}}$ of $\{\theta_n,\, n\in \mathbb{N}\}$ such that $\lim_{k\rightarrow\infty}\|\nabla l_{n_k}(\theta_{n_k})\|=0$.
\end{thm}

The most stringent assumptions on $l$ is simply that its `jumps' are
bounded and that its variations are at most quadratic for large increments of
$\theta$. This condition is always satisfied for $l$ bounded. When
combined with the characterisation of local minima
(\cref{thm:ermoliev-charact-local-min}) in our framework,
\cref{thm: general conv gaussian} constitutes a tool to identify local minima candidates. In
particular, as will shall see, this theorem implies that if the sequence
$\crl{\theta_n,\, n\in \mathbb{N}}$ defined in Algorithm~\ref{alg:main_det_algo}
converges to some $\theta$, then $\theta$ is a candidate local minimum of
$l$. \Cref{thm: general conv gaussian} therefore provides a convergence
result for Algorithm~\ref{alg:main_det_algo} under mild assumptions on $l$, but
it also leads to a number of consequences and stronger results, when more is
known on the objective function. For instance, when $\theta \mapsto l(\theta)$ is convex,
the functions $\theta\mapsto l_n(\theta)$ are also convex for $n\in
	\mathbb{N}$: in this case, one can easily show from Theorem~\ref{thm: general conv gaussian} that Algorithm~\ref{alg:main_det_algo}
converges to the minimiser of $l$.

\subsection{Beyond Gaussian families: mirror descent}\label{sub:beyond_gauss}

Explicitly introducing the Kullback-Leibler minimisation in \Cref{alg:main_det_algo} suggests natural extensions beyond Gaussian families,
and leads to a connection with mirror descent algorithms. As we establish in the Appendix, any
regular EDM (exponential dispersion models) family
\citep{Jorgensen:1987,jorgensen1997theory} could be used, i.e., probability densities of the form:
\[ \pi_{\theta,\gamma}(x) =
	\exp\left\{ \frac{1}{\gamma}\left[\langle\theta,T(x)\rangle-A(\theta)\right]\right\}
	\upsilon_\gamma(x)\,,
\]
where $T\colon\mathsf{X}\rightarrow\mathsf{T}$ and $\theta\in\Theta$. One may
use for instance a Wishart family to perform optimisation when  $\mathsf{X}$ is
a space of symmetric positive definite matrices, or a Bernoulli product
family when  $\mathsf{X}=\{0, 1\}^d$.
Simple derivations (see \cref{app: exponential families,sub:edm_details}) show that, for a
generic EDM family, \eqref{eq:non-homogeneous-descent-generic} becomes
\begin{equation} \label{eq:mirror-descent}
	\nabla_{\theta}A(\theta_{n+1}) =\nabla_{\theta}A(\theta_{n})
	+\gamma_{n}\nabla_{\theta}\log\int\exp\left\{-l(x)\right\}\pi_{(\nabla A)^{-1}(\mu(\theta_{n})),\gamma_{n}}\big({\mathrm d}x\big)\,,
\end{equation}
which is a mirror descent recursion. In the Gaussian case, $T(x)=x$,
$\nabla A(\theta)=\theta$,
and one recovers the   gradient descent recursion~\eqref{eq:intro-moment-match}.
\Cref{sub:wishart-example} provides some details on the Wishart case. 

The motivation and interest behind EDMs in the present context is the following. According to \cref{sub:edm_details} 
one has $\mathbb{E}_{\pi_{\theta,\gamma}}\big(T(X)\big)=\nabla_{\theta}A(\theta)=:\mu(\theta)$ and 
$\mathrm{ var}_{\pi_{\theta,\gamma}}\big(T(X)\big)=\gamma\nabla_{\theta}^{2}A(\theta)$.
Then, for any $\theta\in\Theta$, letting $\gamma\downarrow0$ ensures that
the distribution of $T(X)$ under $\pi_{\theta,\gamma}$ concentrates
on $\mu(\theta)$. In the most common scenario where $T(x)=x$ (or
a component of $T(x)$ is $x$), which is the case of the normal example
we started with, this means that whenever $\theta\mapsto\mu(\theta)$
spans $\mathsf{X}$ then we can aim to adjust $\theta$ to ensure
$\mu(\theta)\in\arg\min_{x\in\mathsf{X}}l(x)$. The use of the symbol $\theta$ instead of $x$ earlier should now
be clear, since their nature is very different in the general scenario, but confounded
in the normal scenario where the mean is the sole parameter used.

\begin{example}
As an illustration of the flexibility of the framework we develop, consider the problem of optimising a function defined on the set of positive symmetric matrices (PSM). Here, rather than using a Gaussian kernel not adapted to the PSM space, we consider a domain-adapted Wishart kernel.
The Wishart kernel builds on the Wishart probability density which, after appropriate rescaling, leads to a member of the EDM family (see Appendix~\ref{sub:wishart-example} for details) $\pi_{\theta,\gamma}$ with potential function $A(\theta)=-1/2\log|-\theta|$, where $|\cdot|$ denotes the determinant of the $d\times d$ positive symmetric matrix $-\theta$.
Using the notation $\Sigma_n=-\theta_n^{-1}/2$, one has $\nabla A(\theta_n)=\Sigma_n$ and  \eqref{eq:mirror-descent} yields:
\begin{equation}
\label{eq:mirror_wishart_summary}
    \Sigma_{n+1} = \Sigma_n - \gamma_n \nabla l_{\gamma_n}\!\big(-\Sigma_n^{-1}/2\big), 
    \qquad n \geq 0,
\end{equation}
while setting the gradient to zero leads to the moment matching solution $\Sigma_{n+1} = \E_{\tilde{\pi}_n}[X]$, with $\tilde{\pi}_n \propto e^{-l}\pi_{\theta_n,\gamma_n}$.

\end{example}

We are currently investigating extensions of our theoretical framework to non-Gaussian families, with particular attention to the Wishart case, and a more general treatment is left for future work.

\subsection{Links to other optimisation schemes}

Algorithm~\ref{alg:main_det_algo} belongs to the family of so-called model-based search methods, which can be traced back to \cite{dorigo1992ottimizzazione,Bonet1996MIMICFO}, appearing implicitly in the former and explicitly in the latter; see \cite{zlochin2004model} for a detailed survey. These ideas have since been rediscovered on multiple occasions, again implicitly or explicitly, and further developed, e.g. \cite{rubinstein1999cross,
	doi:10.1073/pnas.0603181103,
	Ionides2011_IF, 
	osher2019laplacian,
	osher2023,spokoiny:2023}. Earlier contributions were motivated by problems in discrete optimisation. The cross-entropy (CE) method of \cite{rubinstein1999cross}, is closest to ours. Exponential families and Kullback-Leibler minimisation are also suggested, but concentration is obtained by rescaling the objective, similarly to our discussion in Subsection~\ref{subsec:implementation}, while in the Gaussian case both the mean and covariance are estimated \citep{rubinstein2004cross}. In contrast, we impose a specific form for the sequence of covariance matrices, which in turn determines a sequence of step-sizes in the time-inhomogeneous gradient descent reinterpretation of Algorithm~\ref{alg:main_det_algo}. We also identify Exponential Dispersion Models (EDMs) as a particularly well adapted family of parametric models since their variance can be controlled independently of location. To the best of our knowledge, convergence of the CE method for a general class of objective functions is still lacking while there is empirical evidence that it may not always converge; see, e.g., \cite{CETetris} in the context of reinforcement learning.
 Algorithm 1 in \cite{Ionides2011_IF} closely resembles Algorithm~\ref{alg:main_det_algo} in the particular Gaussian case, but generality of the approach and connection to earlier literature seem to have been missed, and the focus is mainly on inference in state-space models. \cite{Ionides2011_IF} also provides a convergence analysis under strong assumptions, in particular assuming that the objective is differentiable.


In recent work \citep{osher2023, tibshirani2024}, a
recursion similar to ours is proposed, albeit with fixed stepsizes. Perspective
of their work is however significantly different. While our algorithm was motivated by Bayes' rule,  the
recursion in the aforementioned papers is obtained by considering a
Gaussian transformation to a Hamilton Jacobi system of partial differential
equations representing the Moreau envelope of a proximal minimisation problem,
and therefore relates to infinitesimal convolutions. \cite{spokoiny:2023} proposes an update similar to ours  in the Gaussian kernel scenario. However, this is where similarities seem to end as the motivation appears slightly different and the analysis of the properties of the algorithm significantly different, in particular requiring differentiability of $l$ and using concentration properties.

Joe Watson (Applied Intelligence Lab of the Oxford Robotics Institute) has pointed out links to several studies in robotics control that illustrate compelling applications of algorithms of the type studied in this work (see, e.g., \cite{watson23a, belousov2018, abdolmaleki_2015, Deisenroth_et_al_2013}).

\subsection{Organisation of the paper}

As mentioned above, our analysis of \cref{alg:main_det_algo}  is  based on its equivalence to a time-inhomogeneous gradient descent, \eqref{eq:non-homogeneous-descent-generic} on smooth approximations of $l$, \eqref{eq:Laplace_gaussian}, which is a particular type of Laplace functionals. In
\cref{sec:Convergence:-deterministic-scena}, we provide conditions for convergence of a general inhomogeneous gradient descent algorithm based on a general sequence of smooth approximations $\crl{f_n:\R^d\to \R; n\geq 0}$ of an objective function $f:\R^d\to \R$. \cref{sec:Laplace_functionals} focuses on the analysis of Laplace functionals, namely smooth approximating functions of the form
\begin{equation}
	\label{eq:laplace_functionals_general_kernel}
	f_{n}(\theta)\coloneq -\log\round{\int \exp\{-f(x)\}\psi_{n,\theta}(\mathrm{d} x)},\quad n\geq 1
\end{equation}
where $\psi_{n,\theta}(\mathrm{d} x)$ is a smoothing kernel; for instance, the Gaussian kernel with mean $\theta$ and variance $\gamma_n$ as in (\ref{eq:Laplace_gaussian}). The results from \cref{sec:Convergence:-deterministic-scena,sec:Laplace_functionals} are not only directly applicable to the analysis of \cref{alg:main_det_algo}, but also hold independent interest and may be valuable for studying a broad range of algorithms. In \cref{sec:convergence_of_algorithm} the results derived in \cref{sec:Convergence:-deterministic-scena,sec:Laplace_functionals} are used to prove \Cref{thm: general conv gaussian} and some numerical experiments are presented in \cref{sec:experiments}. \cref{discussion} concludes.

\section{Results on gradient descent with smooth approximations \label{sec:Convergence:-deterministic-scena}}

\subsection{Overview}

In this section we review and extend essential notions and tools required to address the
minimisation of a function $f:\R^d\to \R$, assumed lower-bounded, not
necessarily differentiable, but for which there exists a sequence of
differentiable approximations $\crl{f_n:\R^d\to\R, \ f_n\in C^1(\R^d), \  n\in
		\mathbb{N}}$, which converges to $f$ in a sense to  be made precise below. In this
scenario, it is natural to suggest the non-homogeneous gradient descent
algorithm
\begin{equation} \label{eq:non-homogeneous-gradient-descent}
	\theta_{n+1} = \theta_n-\gamma_n \nabla f_n(\theta_n) \, ,
\end{equation}
where $\crl{\gamma_n\in \R_+, \ n\in \mathbb{N}}$ with $\gamma_n \downarrow 0$,
in the hope that tracking the sequence points in $\crl{\argmin f_n, n\in
		\mathbb{N}}$ or $\crl{\mathrm{loc-}\argmin f_n, n\in
		\mathbb{N}}$  will lead us to minima or local minima of $f$. As we shall see it is sufficient to focus exposition on global minima, as local minima are global minima of the objective function constrained to a neighbourhood, which will turn out to be sufficient for our purpose. This is however a subtle matter given the generality, as we illustrate below.

To start with, non-differentiable functions
may not have a minimum; see, e.g., the left panel of \cref{fig:examples}.
A weak condition ensuring the existence of minima is that $f$ is lower
semi-continuous (Subsection~\ref{subsec:lsemicontinuity}).

A second issue is that even when perfect minimisation of $f_n$ for all $n\in
	\mathbb{N}$ is possible, the intuitive set-limit $\mathrm{Lim}_n\crl{\argmin f_n}=\argmin f$, properly defined in \cite[Chatper 4]{RockWets98}, may not hold. This
is illustrated on the right panel of \cref{fig:examples}: this counter-example
shows that we may not have $\mathrm{Lim}_n\crl{\argmin f_n}=\argmin f$ even when $f_n$
converges uniformly to $f$.
An important point in the present paper is that using smoothed approximations for
minimisation may not work in certain scenarios.

\begin{figure}[t]
	\centering
	\includegraphics[scale = 0.30]{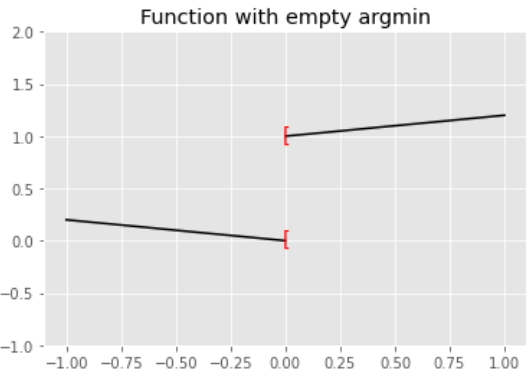}
	\includegraphics[scale = 0.30]{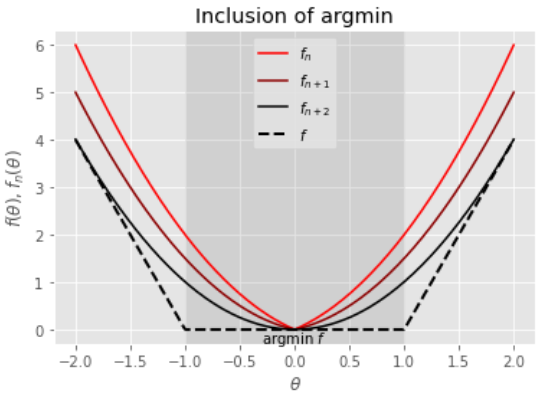}
	\caption{Left: example of a function with empty $\argmin$. Right: A sequence of functions $\crl{f_n, \ n\in \mathbb{N}}$ (red-to-black color palette) that epi-converges to the function $f$. The picture is inspired by \cite[Figure 7.7]{RockWets98}. It is clear that here $\argmin f_n\subset \argmin f$.}
	\label{fig:examples}
\end{figure}

\begin{example}
	Take $d=2$ and consider the
	function
	\[ l(\theta) \coloneq
		\min\{1,|\theta_{1}|\}\ind\{\theta_{1}=\theta_{2}\} + \ind\{\theta_{1}\neq \theta_{2}\}\, ,
	\]
	where here $\theta_1$ and $\theta_2$ are the two components of $\theta$.
	This function is such that  $\mathrm {Leb}\{\theta\in\mathbb{R}^{2}\colon
		l(\theta)<1\}=0$ and, as a consequence, for any  $\theta \in \Theta$ and $\gamma>0$,
	$l_\gamma(\theta)=1$, that is the smoothed functions $\theta\mapsto
		l_\gamma(\theta)$ ``cannot see the minimum'' at zero. In general a requirement
	therefore seems to be that for any $\theta_{*}\in\arg\min_{\theta\in\Theta}l(\theta)$,

	\begin{equation} \label{eq:smooth-gaussian-discriminate}
		\lim_{\gamma\downarrow0}\gamma^{-d/2}
		\int\exp\left\{-l(x)\right\}\phi\Bigl(\frac{x-\theta_*}{\sqrt{\gamma}}\Bigr){\mathrm d}x=
		\exp\left\{-l(\theta_*)\right\}\,.
	\end{equation}
	In this counterexample the left hand side is equal to $\exp(-1)$ while the
	right hand side is equal to $1$.
\end{example}

Epi-convergence of $\crl{f_n:\R^d\to\R,  \  n\in \mathbb{N}}$ to $f$ (properly defined in
\cref{def-epilimits-1}) is a
suitable and flexible form of convergence to establish when smoothing techniques
combined with exact optimisation achieve their goal. Ensuring this property is
a natural prerequisite to the justification of the recursion
\eqref{eq:non-homogeneous-gradient-descent} to optimise $f$; essentials of Epi-convergence are covered in Subsection~\ref{subsec:epi-cv-cv-minimisation} where we also provide a first consequence, Theorem~\ref{thm:argmins-inclusion}, a weaker form of the set limit $\mathrm{Lim}_n\crl{\argmin
		f_n}=\argmin f$.

The following result, due to \citet[Theorem 4.7]{Ermoliev1995}, exemplifies
what one may hope to be able to deduce for the sequence \eqref{eq:non-homogeneous-gradient-descent}. Assume that $\crl{f_n:\R^d\to\R, \
		f_n\in C^1(\R^d), \  n\in \mathbb{N}}$ epi-converges  to $f$, then
\begin{equation} \label{eq:ermoliev-gradient-result}
	\theta_*\in \R^d \ \text{local minimum of }  f \ \Rightarrow \ \exists \crl{\theta_n, \ n\in \mathbb{N}} \ \text{s.t.} \ \theta_n\to \theta_* \ \text{and} \ ||\nabla f_n(\theta_n)||\to 0 \, .
\end{equation}

The practical implication of this result is that if a sequence
$\crl{\theta_n\in \R^d, \  n\in \mathbb{N}}$  admits a subsequence
$\crl{\theta_{n_k}\in\R^d, \ k\in \mathbb{N}}$ convergent to some
$\theta_{*}\in \R^d$ and such that $\lim_k ||\nabla
	f_{n_k}(\theta_{n_k})||=0$, then $\theta_{*}$ is a valid candidate as a
local minimum of $f$, and accumulations points not satisfying the latter
condition must be rejected.
In Subsection~\ref{subsec:epi-cv-cv-minimisation} we establish \Cref{thm:ermoliev-charact-local-min}, a generalization of \citet[Theorem 4.7]{Ermoliev1995} where we relax their assumptions to cover our scenario. The main result of the present section is Theorem~\ref{thm:cv-descent-inho} where we establish  conditions under which the sequence \eqref{eq:non-homogeneous-gradient-descent} admits a subsequence such that  $\|\nabla f_{n_k}(\theta_{n_k})\|\to 0$.

\subsection{Lower semi-continuity} \label{subsec:lsemicontinuity}




To  define semi-continuous functions, we must first define lim inf for
functions.

\begin{definition}
	Let $f:\R^d\to \R$. For $\theta_0\in \R^d$,
	$$\liminf_{\theta\to \theta_0} f(\theta)
		\coloneq \sup_{\epsilon >0}\left[\inf\crl{f(\theta); \ \theta\in \mathbb{B}(\theta_0,\epsilon)}\right]$$
	where $\mathbb{B}(\theta_0, \epsilon)$ denotes a closed metric ball with center $\theta_0$ and radius $\epsilon>0$.
\end{definition}

\begin{definition}[Lower semi-continuity]{\cite[Def. 1.5]{RockWets98}}
	\label{def: lsc}
	A function $f:\R^d\to \R$ is said to be
	\begin{enumerate}
		\item lower semi-continuous ($\mathrm{lsc}$) at $\theta_{0}\in \R^d$ if
		      \begin{equation}
			      \label{eq: lsc}\liminf_{\theta\to \theta_{0}}f(\theta)\geq f(\theta_{0}).
		      \end{equation}
		\item lower semi-continuous if the above holds for any $\theta_{0}\in \R^d$.
	\end{enumerate}
\end{definition}

\begin{remark}
	For $\theta_0\in \R^d$ and $\epsilon>0$, $\inf\crl{f(\theta); \ \theta\in
			\mathbb{B}(\theta_0,\epsilon)}\leq f(\theta_0)$, therefore,
	$\liminf_{\theta\to \theta_0} f(\theta)\leq f(\theta_0)$; hence condition
	\eqref{eq: lsc} is equivalent to $\liminf_{\theta\to \theta_{0}}f(\theta)= f(\theta_{0})$.
\end{remark}

By definition, any lower semi-continuous, lower-bounded function has a minimum on $\R^d$. Upper semi-continuity can be similarly defined, by replacing $\liminf$ with  $\limsup$ and by reversing the inequality.

\begin{definition}[Strong lower semi-continuity]
	\label{def: slsc}
	A function $f:\R^d\to \R$ is said to be
	\begin{enumerate}
		\item strongly lower semi-continuous ($slsc$) at $\theta\in \R^d$
		      if it is lower semi-continuous at $\theta\in \R^d$ and there
		      exists a sequence $\crl{\theta_n, \ n\in \mathbb{N}}$,
		      $\theta_n\to \theta$, with $f$ continuous at every $\theta_n$, and $f(\theta_n)\to f(\theta)$.
		\item strongly lower semi-continuous if the above holds for any $\theta\in \R^d$.
	\end{enumerate}

\end{definition}

In words, strong lower semi-continuity is lower semi-continuity excluding
discontinuities at isolated points. Remark that we do not make any assumption
on smoothness of the function; the class of strongly lower semi-continuous
functions includes indicator functions of closed sets, step functions, ceiling
functions; but also not-everywhere differentiable continuous and discontinuous
(if there are not isolated discontinuity points) functions. Strong upper semi-continuity can be accordingly defined.

\begin{example}
	The indicator function $f(\theta)=\ind\{\theta > 0\}$
	is strongly lower semi-continuous; the indicator function
	$f(\theta) = \ind\{\theta \geq 0\}$ is not
	semi-continuous at $\theta=0$.

	Intuitively, we can note that the epigraph of $f$, namely the part of the space above the graph of $f$, is not a closed set, which in
	fact precludes lower semi-continuity. See \cref{app:lsc_epigraphs} for more details.

\end{example}

\subsection{Epi-convergence and convergence in minimisation} \label{subsec:epi-cv-cv-minimisation}

As discussed in the introduction of this section, epi-convergence is the right
notion to 
formulate and establish convergence of recursions of the type of
\eqref{eq:non-homogeneous-gradient-descent}. We therefore start with some
definitions. A more  classical abstract definition in terms of set convergence of function epigraphs can be provided; see for instance \cite[Chapter 7]{RockWets98}.

\begin{definition}[Epi-convergence] \label{def-epilimits-1}
	A sequence of functions $\crl{f_n:\R^d\to \R, \ n\in\mathbb{N}}$
	epi-converges to a function $f:\R^d\to \R$ if, for each $\theta\in \R^d$,
	\begin{enumerate}
		\item $\liminf_n f_n(\theta_n)\geq f(\theta)$ for any sequence $\theta_n\to \theta$
		\item $\lim_n f_n(\theta_n)=f(\theta)$ for some sequence $\theta_n\to \theta$.
	\end{enumerate}
	Thus, we say that $f$ is the epi-limit of $\crl{f_n, \ n\in \mathbb{N}}$.
\end{definition}

We can immediately note that if a function $f:\R^d\to \R$ is the epi-limit of some sequence $\crl{f_n, \ n\in \mathbb{N}}$, then $f$ is necessarily lower semi-continuous.



\begin{example} \label{exa:epi-limits}
	The three examples below aim to illustrate how epi-convergence differs from
	standard functional limits.  Consider $\crl{f_n\colon \R \rightarrow \R, \
			n\in \mathbb{N}}$ where
	\begin{enumerate}
		\item $f_n(\theta)=-\exp(-n \theta^2)$, then the sequence  epi-converges to
		      $f(\theta) = -\ind\{\theta=0\}$. In fact, it converges pointwise to the
		      same limit. However, epi-convergence generally differs from e.g.
		      pointwise convergence, as the following example shows.

		\item  $f_n(\theta)=\exp(-n \theta^2)$, then the sequence epi-converges to
		      $f(\theta) = 0$, as in particular $\lim_n f_n(n^{-\alpha})=0$ for
		      $\alpha\in(0,1/2)$; we see here how epi-convergence differs from e.g.
		      pointwise convergence since $\lim f_n(0)=1$ here.

		\item $f_n(\theta)=(-1)^n \exp(-n \theta^2)$, then $f_n$ does not
		      epi-converge.
	\end{enumerate}
\end{example}

The following theorem highlights two main consequences of epi-convergence,
describing what type of results one can expect about convergence of minima and
infima. Below, we say that a function $g:\R^d\to \R$ is eventually
level-bounded if for each $a\in \R^d$ the sequence of level-sets
$\crl{\mathrm{lev}_{\leq a}f_n, \ n\in \mathbb{N}}$, defined in Appendix \ref{app:lsc_epigraphs}, is eventually bounded. For
instance, a sequence of eventually lower-bounded functions is eventually
level-bounded.

\begin{thm} [\cite{RockWets98}, Theorem 7.33]
	\label{thm:argmins-inclusion}
	Let
	\begin{enumerate}
		\item $f:\R^d\to \R$ be a lower-bounded, lower semi-continuous function,

		\item $\crl{f_n:\R^d\to \R, \ n\in \mathbb{N}}$ be a sequence of lower semi-continuous functions such that
		      \begin{enumerate}
			      \item $\crl{f_n, \ n\in \mathbb{N}}$ epi-converges to $f$,
			      \item $\crl{f_n, \ n\in \mathbb{N}}$ is eventually level-bounded.
		      \end{enumerate}
	\end{enumerate}

	Then,
	\begin{enumerate}
		\item $\underset{\theta\in \R^d}{\inf}f_n(\theta)\to \underset{\theta\in \R^d}{\inf}f(\theta)$,
		\item $\mathrm{Limsup}_n  \argmin f_n\subset \argmin f$.
	\end{enumerate}
\end{thm}

\citet[Chapter 5]{RockWets98} argue that looking for the inclusion (point 2
above), is preferable over the stronger result $\operatorname{Lim}_n\crl{\argmin f_n}=\argmin
	f$.  Again, refer to \Cref{fig:examples}, where $f_n$ epi-converges
to $f$ (and in fact converges uniformly), but the latter does not hold.


We now turn to results characterising local minima $\theta_*$ of $f$ as
accumulation points of sequences $\crl{\theta_n \in \R^d, n\in\mathbb{N}}$ in
the situation where an epi-convergent and differentiable approximating sequence
$\crl{f_n:\R^d\to \R, \ n\in \mathbb{N}}$ exists. In particular the existence
of $\theta_n \rightarrow \theta_*$ is established, for which $\lim_n \|\nabla
	f_n(\theta_n)\|=0$.  As a consequence for any algorithm producing a sequence
$\crl{\theta_n\in\R^d,n\in\mathbb{N}}$ from which a convergent sequence
$\theta_{n_k} \rightarrow \theta_* \in \R^d$ can be extracted, then, if
$\lim_{n_k} \|\nabla f_{n_k}(\theta_{n_k})\|\neq 0$ we reject $\theta_*$ as a
local minimum.

We begin with two Lemmata that allow us to prove the key result of this section. The following Lemma describes a key differentiability property of the convolution of two functions, one of which is of class $C^1(\R^d)$.

\begin{lemma}[\cite{RockWets98}, Theorem 9.67]
	\label{lemma: differentiability of averages}
	Let $f,g:\R^d\to \R$ be locally integrable and assume that $g\in C^1(\R^d)$. Then the function $\theta\mapsto \int f(u)g(\theta-u)\mathrm{d}u$ belongs to $C^1(\R^d)$ with
	$$\nabla \int f(u)g(\theta-u)\mathrm{d}u=\int f(u)\nabla g(\theta-u)\mathrm{d}u, \ \ \theta\in \R^d.$$

\end{lemma}

The lemma below describes an important property of convergent sequences that attain $\epsilon$-optimal points of epi-convergent functions.

\begin{lemma}[\cite{attouch1984variational}, Theorem 1.10]
	\label{lemma: Ermoliev 3.5}
	Let $D\subseteq \R^d$ and
	\begin{enumerate}
		\item $\crl{f_n:D \to \R, \ n\in \mathbb{N}}$ be epi-convergent to $f:D\to \R$
		\item $\crl{\theta_{n}\in \R^d, \ n\in \mathbb{N}}$ be such that $\theta_{n}\to \theta_{*}\in D$ and for $\crl{\epsilon_n\in \R_+,n \in \mathbb{N}}$ with $\epsilon_n\downarrow 0$
		      $$f_n(\theta_{n})\leq \underset{\theta\in D}{\inf} f_n(\theta)+\epsilon_n, \ n\in \mathbb{N}\, .$$
	\end{enumerate}
	Then, $\theta_{*}\in \underset{\theta\in D}{\argmin} \  f$.
\end{lemma}

Finally, the following justifies an optimisation procedure for finding local minima of $l$ by tracking minimisers of differentiable functions $\crl{l_n, \ n\in \mathbb{N}}$ that epi-converge to $l$. The proof of the result is a generalisation of the proof of \cite[Theorem 4.7]{Ermoliev1995}. While in the latter paper the result is stated for a strongly lower semi-continuous $f$ and auxiliary mollifiers with bounded support (implying epi-convergence in their setting), in our statement epi-convergence is put as an assumption of the theorem and the other assumptions are relaxed: $f$ is allowed to be just lower semi-continuous and the auxiliary mollifiers to be Gaussian. This extends its validity and widens the class of methods whose convergence can be proven in terms of the Theorem below.

\begin{thm} \label{thm:ermoliev-charact-local-min} Let
	\begin{enumerate}
		\item  $f:\R^d\to \R$ be locally integrable, lower bounded and lower semi-continuous,
		\item $\crl{f_n:\R^d\to\R, \ n\in \mathbb{N}}$ be a sequence of
		      differentiable functions epi-convergent to $f$.
	\end{enumerate}
	Then  for any $\theta_* \in \mathrm{loc-}\argmin f$ there exists
	$\crl{\theta_{k} \in \R^d, \  k\in \mathbb{N}}$ such that $\theta_k\to
		\theta_*$ and $\lim_k ||\nabla f_{k}(\theta_{k})||=0$.
\end{thm}

\begin{proof}
	Let $\theta_*$ be a local minimiser of $f$, define $\phi(\theta):=f(\theta)+||\theta-\theta_*||^2$ and let $V$ be a sufficiently small compact set such that $\theta_* \in V$ and $\argmin_{\theta\in V} \ \phi(\theta)=\crl{\theta_*}$ - that is, $\theta_*$ is the unique global minimiser of $\phi$ on $V$. The uniqueness of the minimiser on $V$, ensured by the auxiliary function $\phi$, excludes potential issues arising from working with a locally flat function, and more generally guarantees that a (sub-) sequence converging to $\theta_*$ as per \cref{lemma: Ermoliev 3.5} exists, as we are going to illustrate. Consider the sequence of functions $\crl{\psi_n:\R^d\to \R_+, \ n\in \mathbb{N}}$ defined by $$\psi_n(z):=(2\pi\gamma_n)^{-d/2}\exp\crl{-\frac{||z||^2}{2\gamma_n}}$$ and $\crl{\gamma_n\in \R_+, n\in \mathbb{N}}$ such that $\gamma_n\downarrow 0$.
	From \cite[Remark 3.14]{Ermoliev1995} we have $\psi_n\in C^1(\R^d), \ \ \int\psi_n(z)\mathrm{d}z=1$,  $n\in \mathbb{N}$, and for every $\delta>0$,
	$$\lim_n\int_{\crl{||z||> \delta}}\psi_n(z)\mathrm{d}z=0, \  \text{and} \  \lim_n \int_{||z||>\delta} ||z+\theta-\theta_*||^2\psi_{n}(z)\mathrm{d}z=0$$ uniformly in $\theta\in \R^d$.\\
	Let $\beta_{n}(\theta, \theta_*):=\int ||z+\theta-\theta_*||^2\psi_{n}(z)\mathrm{d}z$, $n\in \mathbb{N}, \  \theta\in V$, and define the auxiliary functions
	$$\phi_n(\theta):=f_n(\theta)+\beta_{n}(\theta, \theta_*), \ n\in \mathbb{N}, \ \theta\in V.$$
	Note that
	\begin{itemize}
		\item the functions $\crl{\beta_{n}, \ n\in \mathbb{N}}$ are of class $C^1(\R^d)$ by \cref{lemma: differentiability of averages},
		\item The sequence $\crl{\beta_n, \ n\in \mathbb{N}}$ epi-converge to $\theta\mapsto ||\theta-\theta_*||^2$ on $V$ due to  \cite[Theorem 3.7 and Remark 3.14]{Ermoliev1995} -- in fact, they converge uniformly,
		\item By assumption, $\crl{f_n, \ n\in \mathbb{N}}$ is a sequence of differentiable functions that epi-converge to $f$; therefore, $\crl{\phi_n, \ n\in \mathbb{N}}$ is a sequence of well-defined differentiable (hence continuous) functions on $V$ that epi-converge to $\phi$ on $V$.
	\end{itemize}
	Let $\crl{\bar{\theta}_n\in V, \ n\in \mathbb{N}}$ be a sequence of minimisers of $\crl{\phi_n, \ n\in \mathbb{N}}$, which exists since for any $n \in \mathbb{N}$, $\phi_n$ is continuous and $V$ a compact set. From the compactness of $V$ there exists a convergent subsequence $\crl{\theta_n \in V, \  n\in \mathbb{N}}$ of $\crl{\bar{\theta}_n\in V, \ n\in \mathbb{N}}$. By \cref{lemma: Ermoliev 3.5}, it holds that $\theta_n\to \theta_*$. We now turn to the second statement. For each $n \in \mathbb{N}$,
	\begin{align}
		0=\nabla \phi_n(\theta_n)=\nabla f_n(\theta_n)+\nabla \beta_{n}(\theta_n, \theta_*).
	\end{align}
	and in the limit,
	\begin{align}
		\label{eq: equality of limit gradients}
		\lim_n\nabla f_n(\theta_n)= -\lim_n \nabla  \beta_{n}(\theta_n, \theta_*)=0.
	\end{align}
	which proves the statement.\\
	In order to show the last equality, by continuity of the function $\theta\mapsto||\theta||^2$ one can use \cref{lemma: differentiability of averages} and write, for every $k$,
	\begin{align*}
		\nabla \beta_k(\theta_k,\theta_*) & = \nabla \int ||z+\theta_k-\theta_*||^2\psi_k(z)\mathrm{d}z \\&\overset{\text{\Cref{lemma: differentiability of averages}}}=\int \nabla ||z+\theta_k-\theta_*||^2\psi_k(z)\mathrm{d}z\\
		                                  & = \int 2(z+\theta_k-\theta_*)\psi_k(z)\mathrm{d}z           \\
		                                  & =2\crl{\int z\psi_k(z)\mathrm{d}z+\theta_k-\theta_*}.
	\end{align*}

	We can now note that, by assumption, $\int z\psi_k(z)\mathrm{d}z=0$ for
	every $k\in \mathbb{N}$; moreover, $\theta_k\to \theta_*$. Therefore, the
	last equality in \eqref{eq: equality of limit gradients} holds.
\end{proof}
\subsection{Convergence of time-inhomogeneous gradient descent} \label{subsec:stab-cv-inhomogeneous}

The following result establishes convergence of time inhomogeneous gradient algorithms to local minima of $l$. This generalizes results such as \cite{gupal_norkin}, not requiring confinement of the sequence $\{\theta_n,n\in\mathbb{N}\}$ for example.

\begin{thm}
	\label{thm:cv-descent-inho}
	For differentiable functions $\{f_n \colon \R^d \rightarrow \R, \ n\in \mathbb{N}\}$ and $\{\gamma_n \in \R_+, \ n\in \mathbb{N}\}$ consider the recursion defined for some $\theta_0 \in \R^d$ and for $n\geq 0$
	\begin{equation} \label{eq:inh-gradient-in-cv-thm}
		\theta_{n+1}=\theta_n-\gamma_n\nabla f_{n}(\theta_n) \,.
	\end{equation}

	Let  $\alpha\in [0,2]$, $\{L_n \in \R_+, \ n\in \mathbb{N}\}$ and $\{\delta_n \in [0,\infty), \ n\in \mathbb{N}\}$ be such that
	\begin{align*}
		\limsup_{n\rightarrow\infty} \gamma_n L_n<1,\quad \lim_{n\rightarrow\infty}(\delta_n/\gamma_n)=\lim_{n\rightarrow\infty}(\delta_n/\gamma_{n+1})=0,\quad\sum_{n=1}^\infty\gamma_n=\infty.
	\end{align*}

	Assume that the following  conditions hold:
	\begin{enumerate}
		\item\label{inf_f} $\inf_{(n,\theta)\in\mathbb{N}\times\R^d} f_n(\theta)>-\infty$.
		\item\label{L-smooth}  for all $\theta,\theta'\in\R^d$ and $n\in\mathbb{N}$,
		      \[f_n(\theta')\leq f_n(\theta)+\langle \nabla f_n(\theta),\theta'-\theta\rangle+L_n\|\theta'-\theta\|^2.\]
		\item\label{time_diff} for all $\theta\in\R^d$ and $n\in\mathbb{N}$,
		      \[f_{n+1}(\theta)-f_{n}(\theta)\leq \delta_n\big[1+ \|\nabla f_{n+1}(\theta)\|^\alpha\big].
		      \]
		\item\label{extra} One of the following conditions holds (with the convention $0\cdot\infty=0$):
		      \begin{enumerate}
			      \item\label{alpha} Condition~\ref{time_diff} holds with $\alpha=0$.
			      \item\label{sup_f}
			            \[
				            \frac{ \delta_{n} \sup_{\theta\in\R^d}\|f_{n+1}(\theta)
				            \|^{\alpha/2}}{\gamma_{n+1}^{\alpha/2}\sum_{m=1}^{n}\gamma_m} \to 0
				            \qquad \mbox{as }n \to \infty.
			            \]
			      \item\label{grad}
			            \begin{align*}
				            \frac{ \delta_{n} \sup_{\theta\in\R^d}\|\nabla f_{n+1}(\theta) \|^\alpha}{\sum_{m=1}^{n}\gamma_m}
				            \to 0\qquad\mbox{as }n\to \infty.
			            \end{align*}
			      \item\label{Lip_f} there exists a constant $\beta\in(0,1]$ and a sequence $\{\tilde{L}_n \in (0,\infty) \colon n\geq 1\}$ such that, for all $\theta,\theta'\in\R^d$ and $n\in\mathbb{N}$, $|f_n(\theta)-f_n(\theta')|\leq \tilde{L}_n\|\theta-\theta'\|^\beta$ and such that
			            \begin{align*}
				            \frac{ \delta_{n} \tilde{L}_{n+1}^\frac{\alpha}{2-\beta}}{\gamma_{n+1}^{\alpha (1-\beta)/(2-\beta)}\sum_{m=1}^{n}\gamma_m}
				            \to 0 \qquad \mbox{as } n \to \infty.
			            \end{align*}
		      \end{enumerate}
	\end{enumerate}
	Then, there exists a subsequence $\crl{\theta_{n_k}, \ k\in \mathbb{N}}$ of $\{\theta_n,n\geq 1\}$ such that $\lim_{k\rightarrow\infty}\|\nabla f_{n_k}(\theta_{n_k})\|=0$.
\end{thm}

\begin{proof}

	Under Condition~\ref{L-smooth} and using \eqref{eq:inh-gradient-in-cv-thm}, for all $n\geq 1$ we have
	\begin{equation}\label{eq:descent}
		\begin{split}
			f_n(\theta_{n+1}) & \leq f_n(\theta_n)+\langle\nabla f_n(\theta_n),\theta_{n+1}-\theta_n\rangle+L_n\|\theta_{n+1}-\theta_n\|^2 \\
			                  & = f_n(\theta_n)-\gamma_n\|\nabla f_n(\theta_n)\|^2\big(1- L_n\gamma_n \big).
		\end{split}
	\end{equation}

	Let $n_1\in\mathbb{N}$ and $\epsilon_1\in(0,1)$ be such that $1-\gamma_n L_n\geq \epsilon_1$ for all $n\geq n_1$. Then, for all $n\geq n_1$ we have, using \eqref{eq:descent} and under
	Condition~\ref{time_diff},
	\begin{equation}\label{eq:main_eq_0}
		\begin{split}
			f_{n+1}(\theta_{n+1}) & \leq f_n(\theta_n)-\gamma_n\|\nabla f_n(\theta_n)\|^2\big(1- L_n\gamma_n \big)+f_{n+1}(\theta_{n+1})-f_n(\theta_{n+1})        \\
			                      & \leq f_n(\theta_n)-\gamma_n \epsilon_1\|\nabla f_n(\theta_n)\|^2+\delta_n\big(1+\|\nabla f_{n+1}(\theta_{n+1})\|^\alpha\big).
		\end{split}
	\end{equation}

	We now prove the result of the theorem by contradiction. To this aim, assume that there exists an $\epsilon_2\in(0,1)$ and an $n_2\in\mathbb{N}$ such that $\|\nabla f_n(\theta_n)\|\geq \epsilon_2$ for all $n\geq n_2$.

	Then,   for all    $n\geq n_3:=\max\{n_1,n_2\}$ we have, by \eqref{eq:main_eq_0} and letting $\epsilon_3=\epsilon_1\epsilon_2^{2-\alpha}$
	\begin{align*}
		f_{n+1}(\theta_{n+1}) & \leq f_{n}(\theta_{n})-\epsilon_3 \gamma_{n}\|\nabla f_{n}(\theta_{n})\|^\alpha+\delta_{n}\big(1+\|\nabla f_{n+1}(\theta_{n+1})\|^\alpha\big).
	\end{align*}
	and thus, for all $n\geq n_3$, we have
	\begin{equation}\label{eq:main_eq}
		\begin{split}
			f_{n+1}(\theta_{n+1}) \leq
			     & f_{n_3}(\theta_{n_3})-\epsilon_3\sum_{m=n_3}^{n}\gamma_m\|\nabla f_m(\theta_m)\|^\alpha                            \\
			     & +\bigg(\sum_{m=n_3}^{n} \delta_{m}  \|\nabla f_{m+1}(\theta_{m+1})\|^\alpha\bigg)+\sum_{m=n_3}^{n} \delta_m        \\
			\leq & f_{n_3}(\theta_{n_3})-\sum_{m=n_3+1}^{n}\|\nabla f_m(\theta_m)\|^\alpha\big(\epsilon_3\gamma_m - \delta_{m-1}\big) \\
			     & + \delta_{n} \|\nabla f_{n+1}(\theta_{n+1})\|^\alpha+\sum_{m=n_3}^{n} \delta_m\,.
		\end{split}
	\end{equation}

	To proceed further assume without loss of generality that $n_3$ is sufficiently large so that, for some $\epsilon_4\in(0,1)$, we have  $\epsilon_3\gamma_n - \delta_{n-1}\geq \epsilon_4\gamma_n$ for all $n\geq n_3$. Then, using \eqref{eq:main_eq}, for all $n\geq n_3$ we have
	\begin{equation}\label{eq:main_eq2}
		\begin{split}
			f_{n+1}(\theta_{n+1}) & \leq f_{n_3}(\theta_{n_3})-\epsilon_4\epsilon_2^\alpha\Big(\sum_{m=n_3+1}^{n}\gamma_m\Big)+ \delta_{n} \|\nabla f_{n+1}(\theta_{n+1})\|^\alpha+\sum_{m=n_3}^{n} \delta_m                                                                      \\
			                      & =f_{n_3}(\theta_{n_3})-\sum_{m=n_3+1}^{n}\gamma_m\bigg(\epsilon_4\epsilon_2^\alpha-\frac{  \delta_{n} \|\nabla f_{n+1}(\theta_{n+1})\|^\alpha}{\sum_{m=n_3+1}^{n}\gamma_m}-\frac{\sum_{m=n_3}^{n} \delta_m}{\sum_{m=n_3+1}^{n}\gamma_m}\bigg)
		\end{split}
	\end{equation}
	where, under the assumptions of the theorem,
	\begin{align*}
		\lim_{n\rightarrow\infty}\frac{\sum_{m=n_3}^{n} \delta_m}{\sum_{m=n_3+1}^{n}\gamma_m}=0,\quad \lim_{n\rightarrow\infty}\sum_{m=n_3+1}^{n}\gamma_m=\infty.
	\end{align*}
	Therefore, if
	\begin{align}\label{eq:toShow}
		\lim_{n\rightarrow\infty}\frac{  \delta_{n} \|\nabla f_{n+1}(\theta_{n+1})\|^\alpha}{\sum_{m=n_3+1}^{n}\gamma_m}=0
	\end{align}
	then, by \eqref{eq:main_eq2}, we have
	$\lim_{n\rightarrow\infty}f_{n+1}(\theta_{n+1})=-\infty$ which contradicts
	Condition~\ref{inf_f}. Hence, to complete the proof it remains to show that \eqref{eq:toShow} holds under the assumption of the theorem.

	Remark first that \eqref{eq:toShow} trivially holds under Condition~\ref{alpha}
	and under Condition~\ref{grad}. Next, remark that, by \eqref{eq:descent}, for
	all $n\geq n_3$ we have
	\begin{align}\label{eq:b}
		f_n(\theta_{n+1})\leq f_n(\theta_n)-\gamma_n\epsilon_1\|\nabla f_n(\theta_n)\|^2\Leftrightarrow \|\nabla f_n(\theta_n)\|^\alpha \leq \bigg(\frac{ f_n(\theta_n)-f_n(\theta_{n+1}) }{ \gamma_n\epsilon_1 }\bigg)^{\frac{\alpha}{2}}
	\end{align}
	from which we readily obtain that \eqref{eq:toShow} holds under Condition~\ref{sup_f}.

	Finally, by \eqref{eq:b}, for all $n\geq n_3$ we have, under
	Condition~\ref{Lip_f} and using~\eqref{eq:inh-gradient-in-cv-thm},
	\begin{align*}
		\|\nabla f_n(\theta_n)\|^\alpha & \leq \bigg(\frac{\tilde{L}_n\| \theta_n-\theta_{n+1}\|^\beta) }{ \gamma_n\epsilon_1 }\bigg)^{\frac{\alpha}{2}}= \bigg(\frac{\tilde{L}_n \gamma_n^\beta\|\nabla f_n(\theta_n)\|^\beta }{ \gamma_n\epsilon_1 }\bigg)^{\frac{\alpha}{2}}                                               \\
		                                & \Leftrightarrow                                                                                                                                                                                                                                                                     \\
		                                & \|\nabla f_n(\theta_n)\|^{\alpha(1-\beta/2)}\leq \bigg(\frac{\tilde{L}_n  }{ \gamma^{1-\beta}_n\epsilon_1 }\bigg)^{\frac{\alpha}{2}} \Leftrightarrow\|\nabla f_n(\theta_n)\|^{\alpha}\leq \bigg(\frac{\tilde{L}_n  }{ \gamma^{1-\beta}_n\epsilon_1 }\bigg)^{\frac{\alpha}{2-\beta}}
	\end{align*}
	and thus \eqref{eq:toShow} follows. The proof of the theorem is complete.
\end{proof}

\section{Laplace functionals}
\label{sec:Laplace_functionals}

We study epi-convergence of sequences of Laplace functionals based on certain kernels possessing a concentration property (called mollifiers below), of which the Gaussian density is a particular case. We also provide a new descent lemma which is based on the definition and construction of Laplace functionals.




\subsection{Epi-convergence of Laplace functionals} \label{subsec:cv-laplace-functionals}

We begin with the definition of
mollifiers.

\begin{definition}[mollifiers]
	\label{def: mollifiers}
	Let $$\crl{\psi_n:\R^d\to \R_+, \   \  \psi_n\in C^1(\R^d), \ \
			\int\psi_n(z)\mathrm{d}z=1, \   n\in \mathbb{N} }$$
	be a sequence of functions such that for every $\delta>0$,
	$$\lim_{n\to \infty} \int_{||z||> \delta}\psi_n(z)\mathrm{d}z=0.$$
	We call such $\psi_n$ mollifiers.
\end{definition}

Gaussian mollifiers are defined as $\psi_n(z) = \gamma_n^{-d/2}
	\phi(\gamma_n^{-1/2} z)$, with $\gamma_n \downarrow 0$, and $\phi$ is the
standard Gaussian density.


\cref{thm: epi convergence direct} shows epi-convergence of mollifier-based Laplace functionals, namely of sequences $\crl{f_n; n\geq 0}$ of the form
\begin{align}
	\label{eq: Laplace functionals}
	f_n(\theta)=-\log \int e^{-l(x)}\psi_n(x-\theta)\mathrm{d}x, \quad \theta \in \R^d.
\end{align}
\begin{remark}
	In the Gaussian scenario, mollifiers in (\ref{eq: Laplace functionals}) can be written as $$\psi_n(x-\theta)=\pi_{\theta,\gamma_n}(x)\coloneq \gamma_n^{-d/2}\phi\round{\gamma_n^{-1/2}\round{x-\theta}}\, ,$$
	for $x,\theta\in \R^d$ and $\gamma_n>0$, which links to the notation in the Introduction and other sections. Here, notation $\psi_n$ is used to highlight the fact that sequences of densities with such concentration property could also include kernels which do not belong to exponential families, as for instance the uniform kernels in \cite{gupal_norkin}.
\end{remark}

It is useful for later calculations to note that by a change of variable
$z=x-\theta$ we can equivalently write $$f_n(\theta)=-\log \int
	e^{-f(\theta+z)}\psi_n(z)\mathrm{d}z.$$

The proof of \cref{thm: epi convergence direct} is inspired by ideas of
\citet[Theorems 3.2, 3.7 and Corollary 3.3]{Ermoliev1995} but it extends their scope.

\begin{thm}
	\label{thm: epi convergence direct}
	Let $f:\R^d\to \R$ be a lower-bounded, strongly lower semi-continuous function. Let $\crl{\psi_n:\R^d\to \R_+, \ n\in \mathbb{N}}$ be mollifiers.
	Let $f_n(\theta)= -\log \int e^{-f(\theta+z)}\psi_n(z)\mathrm{d}z$, $\theta\in \R^d$, $n=1,2, \ldots$. Then, the sequence $\crl{f_n:\R^d\to \R, \  n\in \mathbb{N}}$ epi-converges to $f$.
\end{thm}
\begin{proof}
	Fix $\theta\in \R^d$. We remark that by \cref{prop: epi_closure}, for a
	lower-bounded, integrable function $f:\R^d\to \R$, the epi-closure
	$\mathrm{cl_e} f$ is a lower semi-continuous function and the hypo-closure
	$\mathrm{cl_h} f$ is an upper semi-continuous function. Moreover, it holds
	$\mathrm{cl_e} f\leq f\leq \mathrm{cl_h}f$.  Note also that if a function
	$f$ is lower-bounded and strongly lower semi-continuous, then
	$e^{-f(\cdot)}$ is upper bounded and strongly upper semi-continuous. We break the proof into three steps.
	\begin{itemize}
		\item Define $g:=e^{-f}$ and $g_n(\cdot):=\int e^{-f(\cdot+z)}\psi_n(z)\mathrm{d}z$.  Let $\theta_n\to \theta$. As first step, we show that
		      \begin{equation}
			      \label{eq: closure chain 2}
			      \mathrm{cl_e} g(\theta)\leq \liminf_n g_n(\theta_n)\leq \limsup_n g_n(\theta_n)\leq \mathrm{cl_h} g(\theta)
		      \end{equation}

		      by only using that $f$ is lower-bounded and that $\crl{\psi_n, \ n\in \mathbb{N}}$ is a sequence of mollifiers.\\
		      Fix $\epsilon>0$.     \begin{itemize}
			      \item By upper semi-continuity of $\mathrm{cl_h} e^{-f}$, there exists $\delta=\delta(\epsilon)>0$ such that $$\mathrm{cl_h} e^{-f(\theta+z)}\leq \mathrm{cl_h}e^{-f(\theta)}+\epsilon$$ for all $z\in \R^d$ such that $||z||\leq \delta$.
			      \item For the above $\delta$, by definition of the mollifiers $\crl{\psi_n, \ n\in \mathbb{N}}$, we can choose $n=n(\epsilon,\delta)$ large enough such that, for all $n\geq n(\epsilon,\delta)$, \begin{equation}
				            \label{eq: remainder part 2}
				            0 \leq \int_{||z||> \frac{\delta}{2}}e^{-f(\theta+z)}\psi_n(z)\mathrm{d}z\leq\underset{u\in \R^d}{\sup} e^{-f(u)}\int_{||z||\geq\frac{\delta}{2}} \psi_n(z)\mathrm{d}z\leq   \frac{\epsilon}{2}.
			            \end{equation}
		      \end{itemize}
		      We first show the last inequality in equation \eqref{eq: closure chain 2}.  Let $\delta=\delta(\epsilon)$ as above. For all $n\geq n(\epsilon, \delta)$ large enough such that $||\theta_n-\theta||\leq \frac{\delta}{2}$, we have  $||\theta_n-\theta+z||\leq \delta$ for any $z$ such that $||z||\leq \frac{\delta}{2}$. So we can write
		      \begin{align*}
			      g_n(\theta_n) & =\int e^{-f(\theta_n+z)}\psi_n(z)\mathrm{d}z                                                                                                 \\
			                    & =\int_{||z||\leq \frac{\delta}{2}}e^{-f(\theta_n+z)}\psi_n(z)\mathrm{d}z+\int_{||z||>\frac{\delta}{2}}e^{-f(\theta_n+z)}\psi_n(z)\mathrm{d}z \\
			                    & \leq \int_{||z||\leq \frac{\delta}{2}}\mathrm{cl_h}e^{-f(\theta_n+z)}\psi_n(z)\mathrm{d}z+\frac{\epsilon}{2}                                 \\
			                    & =\int_{||z||\leq \frac{\delta}{2}}\mathrm{cl_h}e^{-f(\theta+\theta_n-\theta+z)}\psi_n(z)\mathrm{d}z+\frac{\epsilon}{2}                       \\
			                    & \leq (\mathrm{cl_h}e^{-f(\theta)}+\epsilon)\int_{||z||\leq \frac{\delta}{2}}\psi_n(z)\mathrm{d}z+\frac{\epsilon}{2}
		      \end{align*}
		      Hence, for $n$ large enough we have
		      \begin{align*}
			      g_n(\theta_n)\leq \mathrm{cl_h}e^{-f(\theta)}+\epsilon+\frac{\epsilon}{2}
		      \end{align*}
		      so for any $\epsilon>0$
		      \begin{align*}
			      \limsup_n g_n(\theta_n)\leq \mathrm{cl_h}g(\theta)+\frac{3}{2}\epsilon.
		      \end{align*}

		      Therefore, by taking $\epsilon\to 0$, we obtain
		      \begin{equation}
			      \label{eq: limsup and hypo-closure}
			      \limsup_n g_n(\theta_n)\leq \mathrm{cl_h}g(\theta).
		      \end{equation}
		      The middle inequality in \eqref{eq: closure chain 2} is obvious, while the first can be proven in a very similar way, using positivity from equation (\ref{eq: remainder part 2}) and the fact that, for $n$ large-enough, $\int_{||z||\leq \frac{\delta}{2}}\psi_n(z)\mathrm{d}z\geq 1-\frac{\epsilon}{2}$.

		\item The next step is to show that $g_n$ hypo-converges to $g$. Here we use strong lower semi-continuity of $f$.

		      \begin{itemize}
			      \item As $g=e^{-f}$ is upper semi-continuous, it holds that $g=\mathrm{cl_h}g$ by \cref{prop: epi_closure}. Hence, by equation (\ref{eq: limsup and hypo-closure}), for any $\theta\in \R^d$ and any sequence $\crl{\theta_n, \ n\in \mathbb{N}}$ with $\theta_n\to \theta$, we have
			            $$\limsup_n g_n(\theta_n)\leq g(\theta).$$
			      \item It remains to show that for any $\theta\in \R^d$, $\lim_n g_n(\theta_n)=g(\theta)$ for at least one sequence such that $\theta_n\to \theta$. Here we use \cref{lemma: second condition of epi convergence}.
		      \end{itemize}
		\item We finally show that $\crl{f_n, \ n\in \mathbb{N}}$ epi-converges to $f$, based on the above results, continuity and monotonicity of $\log$, and on the fact that, by definition, if $\crl{f_n:\R^d\to \R, \ n\in \mathbb{N}}$ is a sequence of functions that hypo-converge to $f:\R^d\to \R$, then $\crl{-f_n, \ n\in \mathbb{N}}$ epi-converge to $-f$.
		      In detail:  by continuity and monotonicity of $\log(\cdot)$ and by hypo-convergence of $g_n$ to $g$, we can first show that $\crl{\log(g_n), \ n\in \mathbb{N}}$ hypo-converges to $\log(g)=-f$. Indeed, the following chain of inequalities holds for any sequence $\crl{\theta_n, \ n\in \mathbb{N}}$ with $\theta_n\to \theta$:
		      \begin{align*}
			      \limsup_n \log\round{g_n(\theta_n)} & =\lim_n \round{\sup_{m\geq n}\log(g_m(\theta_m))}
			      \leq \lim_n \log\round{\sup_{m\geq n}g_m(\theta_m)}                                                                              \\
			                                          & =\log\round{\lim_n\round{\sup_{m\geq n}g_m(\theta_m)}}=\log\round{\limsup_n g_n(\theta_n)} \\
			                                          & \leq \log(e^{-f(\theta)})=-f(\theta)
		      \end{align*}
		      Hence the first condition for hypo-convergence of
		      $\crl{\log(g_n), \ n\in \mathbb{N}}$ to $\log(g)=-f$ holds. For the second condition, we just use hypo-convergence of $\crl{g_n, \ n\in \mathbb{N}}$ to $g$ and continuity of $\log$. Note that the above chain of inequalities would hold for any non-decreasing continuous transformation. Finally, as $f_n=-\log(g_n)$, we conclude that $\crl{f_n, \ n\in \mathbb{N}}$ epi-converge to $f$.

	\end{itemize}

\end{proof}

Compared to the results by Ermoliev, we can work with mollifiers with unbounded support with one less assumption: Ermoliev requires that for any $\delta>0$ $\lim_n\int_{\crl{||z||>\delta}} |f(\theta+z)|\psi_n(z)\mathrm{d}z=0$ uniformly in $\theta\in \R^d$, to control the tail behaviour. Here we can avoid an assumption of this kind as the mollifiers weight the function $x\mapsto e^{-f(x)}$, which is upper-bounded when $x\mapsto f(x)$ is lower-bounded.\\

When the objective is continuous, we obtain stronger convergence results.

\begin{lemma}
	\label{lemma: continuous convergence mollifiers}
	Under the conditions of \cref{thm: epi convergence direct}, if $f$ is also continuous, then the sequence $\crl{f_n, \ n\in \mathbb{N}}$ converges continuously to $f$, that is, $\lim_n f_n(\theta_n)=f(\theta)$ for any sequence $\crl{\theta_n, \ n\in \mathbb{N}}$ such that $\theta_n\to \theta$, for any $\theta\in \R^d$. This also implies that the sequence converges uniformly to $f$ on compact subsets of $\R^d$.
\end{lemma}
\begin{proof}
	For a continuous function $f:\R^d\to \R$, $\mathrm{cl_e}f=\mathrm{cl_h}f=f$ by \cref{prop: epi_closure}. Hence the statement about continuous convergence follows by the same steps that lead to equation (\ref{eq: closure chain 2}) in Theorem \ref{thm: epi convergence direct}, combined with continuity of $\log(\cdot)$. The statement about uniform convergence on compact sets follows by \cite[Theorem 7.14]{RockWets98}.
\end{proof}

\begin{lemma}
	\label{lemma: second condition of epi convergence}
	Let $f:\R^d\to \R$ be a lower-bounded, strongly lower semi-continuous function. Let $$\crl{\psi_n:\R^d\to \R_+, \ \int\psi_n(z)\mathrm{d}z=1, \ n\in \mathbb{N}}$$
	be a sequence of functions such that, for any $\delta > 0$,
	$\lim_n\int_{||z||\geq \delta}\psi_n(z)\mathrm{d}z=0.$ Let $g(\theta):=e^{-f(\theta)}$ and $g_n(\theta):=\int e^{-f(\theta+z)}\psi_n(z)\mathrm{d}z$, $\theta\in \R^d$, $n=1,2,\ldots$. Then, for any $\theta\in \R^d$, there is at least one sequence $\crl{\theta_n, \ n\in \mathbb{N}}$ such that
	$$\lim_n g_n(\theta_n)=g(\theta).$$
\end{lemma}
\begin{proof}
	The proof is given in  \cite[Theorem 3.7]{Ermoliev1995}. Let $\bar{g}:=-g$ and $\bar{g}_n:=-g_n$. We already now that $\bar{g}$ is strongly lower semi-continuous and that
	$\liminf_n\bar{g}_n(\theta_n)\geq \bar{g}(\theta)$ for all sequences $\crl{\theta_n, \ n\in \mathbb{N}}$ such that $\theta_n\to \theta$.

	We show $\lim_n \bar{g}_n(\theta_n)=\bar{g}(\theta)$ for at least one sequence $\theta_n\to \theta$. Here we use \textit{strong} lower semi-continuity and proceed as follows. By strong lower semi-continuity of $\bar{g}$, there exists a sequence $\theta_k\to \theta$ such that $\lim_k \bar{g}(\theta_k)= \bar{g}(\theta)$, with $\bar{g}$ continuous at $\theta_k$.  \Cref{lemma: continuous convergence mollifiers} states that under (local) continuity we have that, for all $k\geq 0$,
	\begin{equation}
		\label{eq: local continuity}
		\lim_n \bar{g}_n(\theta_k)= \bar{g}(\theta_k)
	\end{equation} Now note that set $S:=\crl{\bar{g}(\theta_k); \ k\in \mathbb{N}}$ is such that $S\in \Liminf_n S_n$,  $S_n:=\crl{\bar{g}_n(\theta_k); \ k\in \mathbb{N} }$ where we recall that $S\in \Liminf_n S_n$ consists of all limit points of sequences $\crl{\alpha_n, \ n\in \mathbb{N}}$ with $\alpha_n\in S_n$. By definition, $\Liminf_n S_n$ is closed and, moreover, $\bar{g}(\theta)\in \text{cl}( S)$, where $\text{cl}(\cdot)$ denotes the closure of a set. This means that there exists a sequence $\crl{\alpha_n, \ n\in \mathbb{N}}$ s.t. $\lim_n \alpha_n= \bar{g}(\theta)$ with $\alpha_n\in S_n$. Then let $\theta_n$ be such that $\bar{g}_n(\theta_n)=\alpha_n$, $n=1,2,\ldots$ to obtain the result. Since $\bar{g}=-g$ and $\bar{g}_n=-g_n$, $n=1,2,\ldots$,  the result translates to $\crl{g_n, \ n\in \mathbb{N}}$ and $g$.
\end{proof}



\subsection{Descent lemma for Laplace functionals}
\label{sub:descent_lemma_laplace_exp_fam}

We show that a crucial convexity property from the theory of exponential families allows us to derive a descent lemma for Laplace functionals.
The descent lemma is stated here in full generality, namely for Laplace functionals based on smoothing kernels in the exponential family. This is a result of independent interest, and it also enforces the validity of a descent lemma in the more general scenarios mentioned in Subsection \ref{sub:beyond_gauss}. The result for the Gaussian algorithm is then recovered as a simple by-product. The derivation of the lemma connects the exponential family framework to the idea of \textit{relative smoothness} presented in \cite{bolte2018} and \cite{lu_freund_nesterov_2017}.

\begin{definition}[Bregman divergence]
	Let $u:\R^d\to \R$ be a differentiable function. The Bregman divergence $\mathrm{D}_u:\R^d\times \R^d\to \R$ associated with $u$ is defined as
	$$\mathrm{D}_{u}(x,y)=u(x)-u(y)-\inner{\nabla u(y)}{x-y}, \ x,y\in\R^d.$$
\end{definition}
\begin{properties}\nonumber
	\label{properties: linearity and positivity of Bregman}
	For any $x,y \in \R^d$,
	\begin{enumerate}
		\item For any pair of differentiable functions $u_1,u_2:\R^d \to \R$, $$\mathrm{D}_{\alpha u_1+\beta u_2}(x,y)=\alpha \mathrm{D}_{u_1}(x,y)+\beta \mathrm{D}_{u_2}(x,y).$$
		\item For any differentiable, convex function $u:\R^d\to \R$, we also have that $$\mathrm{D}_u(x,y)\geq 0$$ with $\mathrm{D}_u(x,y)=0$ iff $x=y$.
	\end{enumerate}
\end{properties}

We can now state the general descent lemma.

\begin{thm}[Descent lemma]
	\label{thm: descent lemma general}
	Consider a function $f:\R^d\to \R$ and for $\gamma>0$, $\theta\in \R^d$
	consider an exponential model (see \cref{app: exponential families})
	$$\pi_{\theta,\gamma}(x):=\exp\crl{\gamma^{-1}\sqrd{\inner{\theta}{T(x)}-A(\theta)}}\upsilon_{\gamma}(x)$$
	with sufficient statistic $T:\R^d\to \R^d$, log-partition function $A:\R^d\to \R$ and baseline probability measure $\upsilon_{\gamma}$. Let $$f_{\gamma}(\theta):=-\log\int \exp(-f(x))\pi_{\theta, \gamma}(x)\mathrm{d}x, \ \theta\in \R^d$$  and assume $\int \exp(-f(x))\pi_{\theta,\gamma}(x)\mathrm{d}x<\infty$ for any $\theta\in \R^d$. Then, for any $\theta,\theta'\in \R^d$, it holds that
	\begin{equation}
		f_{\gamma}(\theta')\leq f_{\gamma}(\theta)+\inner{\nabla f_{\gamma}(\theta)}{\theta'-\theta} + \frac{1}{\gamma}\mathrm{D}_A (\theta',\theta),
	\end{equation}
\end{thm}
\begin{proof}
	Let $x,\theta\in \R^d$ and $\gamma>0$. Consider the exponential model from the assumption \begin{equation*}
		\pi_{\theta, \gamma}(x)=\exp\crl{\frac{1}{\gamma}\sqrd{\inner{\theta}{T(x)}-A(\theta)}}\upsilon_{\gamma}(x)
	\end{equation*}
	and let $$\tilde{\pi}_{\theta, \gamma}(x)\propto \exp\big(-f(x)\big)\pi_{\theta, \gamma}(x).$$ Note that $$\tilde{\pi}_{\theta, \gamma}(x)=\exp\big(-f(x)\big)\exp\crl{\frac{1}{\gamma}\sqrd{\inner{\theta}{T(x)}-A(\theta)}+f_{\gamma}(\theta)}\upsilon_{\gamma}(x) \ ,$$
	that is, the distribution $\tilde{\pi}_{\theta, \gamma}$ still belongs to the regular (in the sense of  \cref{def: regular})  exponential family, with log-partition function given by $$\Tilde{A}(\theta)=A(\theta)-\gamma f_{\gamma}(\theta).$$ The Bregman divergence $\mathrm{D}_{\tilde{A}}$ is well-defined as both $\theta\mapsto A(\theta)$ and $\theta\mapsto f_{\gamma}(\theta)$ are differentiable functions.
	By \cite[Proposition 3.1]{wainwrightjordan2008}, $\tilde{A}(\theta) \ \text{is convex}$.  From the properties of the Bregman divergence, we note that the convexity of $\tilde{A}$ implies that for all $\theta,\theta'\in \R^d$
	\begin{equation}
		\mathrm{D}_{\tilde{A}}(\theta',\theta)\geq 0.
	\end{equation}
	By the linearity property, we have
	\begin{equation*}
		0\leq \mathrm{D}_{\tilde{A}}(\theta',\theta) = \mathrm{D}_{A-\gamma f_{\gamma}}(\theta',\theta)=\mathrm{D}_{A}(\theta',\theta)-\gamma \mathrm{D}_{f_{\gamma}}(\theta',\theta)
	\end{equation*}
	and using the definition of the Bregman Divergence, one obtains
	\begin{equation*}
		\gamma\crl{f_{\gamma}(\theta')-f_{\gamma}(\theta)-\inner{\nabla f_{\gamma}(\theta)}{\theta'-\theta}}\leq \mathrm{D}_A (\theta',\theta).
	\end{equation*}
	That is,
	\begin{equation*}
		f_{\gamma}(\theta')\leq f_{\gamma}(\theta)+\inner{\nabla f_{\gamma}(\theta)}{\theta'-\theta} + \frac{1}{\gamma}\mathrm{D}_A (\theta',\theta),
	\end{equation*}
	which concludes the proof.
\end{proof}

In the Gaussian case, as a corollary, we recover a standard descent lemma in terms of the Euclidean distance.

\begin{corollary}
	\label{cor: descent lemma gaussian}
	Consider a function $f:\R^d\to \R$. For $\gamma>0$, let $f_{\gamma}(\theta):=-\log\int \exp(-f(x))\phi\round{\frac{x-\theta}{\sqrt{\gamma}}}\mathrm{d}x$, $\theta\in \R^d$ and assume $\int \exp(-f(x))\phi\round{\frac{x-\theta}{\sqrt{\gamma}}}\mathrm{d}x<\infty$ for any $\theta\in \R^d$. Then, for any $\theta,\theta'\in \R^d$, it holds that
	\begin{equation}
		f_{\gamma}(\theta')\leq f_{\gamma}(\theta)+\inner{\nabla f_{\gamma}(\theta)}{\theta'-\theta}+\frac{1}{2\gamma}||\theta'-\theta||^2.
	\end{equation}
\end{corollary}
\begin{proof}
	This follows directly from \cref{thm: descent lemma general} by noting that, for $\pi_{\theta,\gamma}(\cdot):=\gamma^{-d/2}\phi\round{\frac{\cdot-\theta}{\sqrt{\gamma}}}$, one has $A(\theta)=\frac{||\theta||^2}{2}$ and $\mathrm{D}_{A}(\theta,\theta')=\frac{1}{2}||\theta'-\theta||^2$, for any $\theta,\theta'\in \R^d$.
\end{proof}


\section{Proof of \cref{thm: general conv gaussian}}
\label{sec:convergence_of_algorithm}

We are now ready to go back to \cref{alg:main_det_algo} and note here that the Laplace functionals based on Gaussian mollifiers, defined for $(\theta,\gamma)\in \R^d\times \R_+$ as

$$l_{\gamma}(\theta)\coloneq -\log\round{\int e^{-l(x)}\gamma^{-d/2}\phi\round{\gamma^{-1/2}(x-\theta)}\mathrm{d}x}\, ,$$

satisfy the conditions of
\cref{thm:cv-descent-inho} with $f_n=l_{\gamma_n}$ for some sequence $\gamma_n\downarrow 0$; specifically,
Assumption~\ref{L-smooth} follows by \cref{cor: descent lemma gaussian}
and Assumption~\ref{time_diff} by \cref{lemma: time inhom condition},
which we prove below. Hence, convergence of recursion (\ref{eq:non-homogeneous-descent-generic}) holds by \cref{thm:cv-descent-inho}, and we write the specific statement in \cref{thm: general conv gaussian}. By the equivalence, this provides us with a convergence result for \cref{alg:main_det_algo}. The interpretation of the result in \cref{thm: general conv gaussian} can be phrased in terms of epi-convergence of
Laplace functionals (\cref{thm: epi convergence direct}), with the interpretation being provided and
justified by our characterisation of local minima under epi-convergence
(\cref{thm:ermoliev-charact-local-min}).\\

In the proof below we specifically denote $\psi(z):=\gamma^{-d/2}\phi(\gamma^{-1/2}z)$, $z\in \R^d$, $\gamma>0$, the multivariate Gaussian with covariance $\gamma \mathbf{I}_d$, $\gamma>0$, and zero mean. This notation is consistent with the mollifiers' framework presented in subsection \cref{subsec:cv-laplace-functionals}.

\begin{lemma}
	\label{lemma: time inhom condition}
	Let   $\{\gamma_n,\,n\in\mathbb{N}\}$ be a sequence on $(0,\infty)$ such that $\lim_{n\rightarrow\infty}\gamma_n=0$ and such that  $\gamma_{n+1}\leq \gamma_n$ for all $n\geq 1$,  and for all $n\in\mathbb{N}$ let $\psi_n(z)=\gamma_n^{-d/2}\phi(\gamma_n^{
			-1/2}z)$. Assume that there exists a constant $C_l\in(0,\infty)$ such that $|l(\theta')-l(\theta)|\leq C_l+C_l\|\theta-\theta'\|^2$ for all $\theta,\theta'\in\R^d$ and let
	\begin{align*}
		\delta_n=\big((\gamma_n/\gamma_{n+1})^{d/2}-1\big)(\gamma_n+1)+(\gamma_n-\gamma_{n+1})+\gamma_n^2,\quad\forall n\geq 1.
	\end{align*}
	Then, there exists a constant $\bar{C}\in(0,\infty)$ and an $n'\in\mathbb{N}$ such that
	\begin{align*}
		\sup_{\theta\in\R^d} \big(l_{n+1}(\theta)-l_n(\theta )\big)\leq \bar{C}\delta_n,\quad\forall n\geq n'.
	\end{align*}
\end{lemma}

\begin{remark}
	If $\gamma_n=n^{-\beta}$ for all $n\geq 1$ and some $\beta\in (0,1)$. Then,    $\sum_{n\geq 1}\gamma_n=\infty$ and, since we have $(\gamma_n/\gamma_{n+1})^{d/2}-1\approx n^{-1}$ and $\gamma_n-\gamma_{n+1}=n^{-\beta-1}$, it follows that $\delta_n=\smallo(\gamma_n)$   and thus $\delta_n/\gamma_{n}\rightarrow 0$.
\end{remark}

\begin{proof}

	Let $n\geq 1$ and $\theta\in\R^d$. If $g_n(\theta)\leq g_{n+1}(\theta)$ we have $l_{n+1}(\theta)-l_n(\theta )\leq 0$ and thus below we assume that $g_{n+1}(\theta)\leq g_n(\theta)$. Then, using the fact that for any real numbers  $0<x<y$ we have $ \log(y)-\log(x)\leq (y-x)/x$, it follows that
	\begin{align}\label{eq:b1}
		l_{n+1}(\theta)- l_n(\theta) & \leq \frac{g_n(\theta)-g_{n+1}(\theta)}{g_{n+1}(\theta)}.
	\end{align}

	To proceed further let $c_n= (\gamma_n/\gamma_{n+1})^{d/2}\geq 1$ and note that
	\begin{align*}
		0\leq \frac{\psi_{n+1}(x-\theta)}{\psi_n(x-\theta)}\leq c_n,\quad\forall x\in\R^d.
	\end{align*}

	In addition, let $n'\in\mathbb{N}$ be such that $\gamma_m\leq 1/(4C_l)$ for all $m\geq n'$, with $C_l$ as in the statement of the lemma. Then, assuming that $n\geq n'$, we have
	\begin{equation}\label{eq:b2}
		\begin{split}
			g_n(\theta)-g_{n+1}(\theta) & =\int e^{-l(x)}\big(\psi_n(x-\theta)-\psi_{n+1}(x-\theta)\big)\mathrm{d} x                                                                  \\
			                            & =e^{-l(\theta)}\int e^{-(l(x)-l(\theta))}\big(\psi_n(x-\theta)-\psi_{n+1}(x-\theta)\big)\mathrm{d} x                                        \\
			                            & =e^{-l(\theta)}\int e^{-(l(x)-l(\theta))}\Big(1-\frac{\psi_{n+1}(x-\theta)}{\psi_n(x-\theta)}\Big)\psi_n(x-\theta)\mathrm{d} x              \\
			                            & \leq  e^{-l(\theta)}\int e^{-(l(x)-l(\theta))}\Big(c_n-\frac{\psi_{n+1}(x-\theta)}{\psi_n(x-\theta)}\Big)\psi_n(x-\theta)\mathrm{d} x       \\
			                            & \leq e^{-l(\theta)+C_l}\int e^{C_l\|x-\theta\|^2}\Big(c_n-\frac{\psi_{n+1}(x-\theta)}{\psi_n(x-\theta)}\Big)\psi_n(x-\theta) \mathrm{d} x   \\
			                            & = e^{-l(\theta)+C_l}\Big(c_n\int e^{C_l\|x-\theta\|^2}\psi_n(x-\theta)\mathrm{d} x-\int e^{C_l\|x-\theta\|^2}\psi_{n+1}(x-\theta)\mathrm{d} x\Big) \\
			                            & =e^{-l(\theta)+C_l}\Big( c_n(1-2C_l\gamma_n)^{-d/2}-(1-2C_l\gamma_{n+1})^{-d/2}\Big).
		\end{split}
	\end{equation}

	Using Taylor's theorem, there exists a constant $C<\infty$ such that
	\begin{align*}
		c_n(1-2C_l\gamma_n)^{-d/2}-(1-2C_l\gamma_{n+1})^{-d/2} & \leq c_n\big(1+dC_l\gamma_n+C\gamma_n^2 \big)-\big(1+dC_l\gamma_{n+1}-C \gamma_{n+1}^2 \big) \\
		                                                       & =(c_n-1)\big(dC_l\gamma_n+1\big)+dC_l(\gamma_n-\gamma_{n+1})+C(\gamma_n^2+\gamma_{n+1}^2)    \\
		                                                       & \leq \max\{1,dC_l, 2C\}\Big((c_n-1)(\gamma_n+1)+(\gamma_n-\gamma_{n+1})+  \gamma_n^2\Big)    \\
		                                                       & =\max\{1,dC_l, 2C\}\delta_n
	\end{align*}
	which, together with \eqref{eq:b1}-\eqref{eq:b2} and letting $C'=e^{C_l}\max\{1,dC_l, 2C\}$,  shows that
	\begin{align}\label{eq:b3}
		l_{n+1}(\theta)- l_n(\theta) & \leq \delta_n \frac{ C' e^{-l(\theta)}}{g_{n+1}(\theta)}.
	\end{align}

	On the other hand,
	\begin{equation}\label{eq:b4}
		\begin{split}
			g_{n+1}(\theta) & = \int \exp(-l(\theta+\gamma_{n+1}^{\frac{1}{2}} z))\phi(z)\mathrm{d} z                                           \\
			                & =e^{-l(\theta)} \int \exp\Big(-\big(l(\theta+\gamma_{n+1}^{\frac{1}{2}} z)-l(\theta)\big)\Big)\phi(z)\mathrm{d} z \\
			                & \geq e^{-l(\theta)-C_l}\int \exp(-C_l\gamma_{n+1} z^2)\phi(z)\mathrm{d} z                                         \\
			                & =e^{-l(\theta)}\Big(e^{-C_l}(1+2C_l\gamma_{n+1})^{-1/2}\Big)                                               \\
			                & \geq  e^{-l(\theta)} e^{-C_l}(1+2C_l\gamma_{1})^{-1/2}
		\end{split}
	\end{equation}
	where the last inequality uses the fact that the sequence $\{\gamma_n,\,n\in\mathbb{N}\}$ is assumed to be non-increasing.

	By combining \eqref{eq:b3} and \eqref{eq:b4}, we obtain

	\begin{align*}
		l_{n+1}(\theta)- l_n(\theta) & \leq  \bar{C}\delta_n,\quad\bar{C}=C' e^{C_l}(1+2C_l\gamma_{1})^{1/2}
	\end{align*}
	and the proof of the lemma is complete.

\end{proof}


\section{Experiments\label{sec:experiments}}

In this section we provide implementation details and evaluate our
methodology on a statistical problem arising in machine learning.

\subsection{Implementation} \label{subsec:implementation}

\Cref{alg:main_det_algo} relies on theoretical distributions $\pi_{n}$ and
$\tilde{\pi}_{n}$ which are typically intractable. We may approximate $\pi_{n}$ with a
Monte Carlo sample of size $N$, i.e.~$X_n^i\sim \pi_n$ for
$i=1,\dots,N$.  Then, since $\tilde{\pi}_{n+1}(x) \propto \pi_n(x)
	\exp\{-l(x)\}$, we may approximate $\tilde{\pi}_{n+1}$ through importance
sampling, that is, assign weight $w_n^i \gets \exp\{-l(X_n^i)\}$ to particle
$X_n^i$;  see \cref{alg:MC_version_det_algo} in
\cref{app:additional-sim-results} for more details.

To reduce the variability of the output, we may also use randomised quasi-Monte
Carlo  to generate the $X_n^i$, see \cref{alg:QMC_version} in
\cref{app:additional-sim-results}.  We use the latter
algorithm below. We observe that it works well even for low values of
$N$; we set $N=128$ throughout.  We define the output of the algorithm at
iteration $n$ to be $\arg\min_{i}l(X_{n}^{i})$, i.e., the particle with the
smallest observed value for $l(x)$.

Since minimising $l$ is equivalent to minimising $l_{\lambda}:=\lambda\times l$
for any scalar $\lambda>0$, one may consider different strategies to scale $l$
automatically to improve speed of convergence. We found the following
approaches to work well in practice: set the scale $\lambda$ so that the
variance of the log-weights equals one, either for the first $k$ iterations,
or for all iterations.  The results reported below correspond to the latter
strategy. We note that, currently, the theoretical framework developed in this
manuscript corresponds to the algorithm used in
\cref{app:additional-sim-results} and do not yet cover the algorithm used
below.

\subsection{AUC scoring and classification}

We illustrate our methodology on a staple machine learning scoring and
classification task. Given training data $\{(Z_{i},Y_{i})\in
	\mathbb{R}^{p}\times\{-1,1\}\colon i=1,\ldots,n_\mathrm{data}\}$, assumed to
arise from a probability distribution $\mathbb{P}$,  we wish to construct a
score function $s\colon \mathbb{R}^{p} \rightarrow \{-1,1\}$, such that for two
independent realisations $(Z,Y)$ and $(Z',Y')$ the theoretical quantity
\[
	\mathbb{P}\big(\left[ s(Z)-s(Z')\right] (Y-Y')<0\big)
\]
is as small as possible. This quantity
is often called the area under curve (AUC) risk function; one of the
motivations for this criterion is that it is less sensitive to class imbalance
than other more standard classification criteria.

Assuming further a particular parametric form for $s(z)$, e.g. $s(z)=s_\theta(z)=\theta^{\top}z$ for
$\theta \in \mathbb{R}^p$,
\citet{ustats} proposed to estimate $x$ through empirical risk
minimisation, i.e.
\[
	\hat{\theta}=\arg\min_{\theta\in\mathbb{R}^{p}}u(x)\,,
\]
where $u(\theta)$ is the following U-statistic:
\begin{equation}
	u(\theta) = \frac{1}{n_\mathrm{data}(n_\mathrm{data}-1)}
	\sum_{i,j=1}^{n_\mathrm{data}}\mathbf{1}\left\{\big[ s_{\theta}(Z_{i})-s_\theta(Z_{j})\big] (Y_{i}-Y_{j})<0\right\}.\label{eq:def_auc}
\end{equation}

This function is challenging to minimise directly, for two reasons:
(a) it is piecewise constant and therefore discontinuous; (b) it is
invariant by affine transformations for the linear model, i.e. $u(\lambda \theta)=u(\theta)$
for any scalar $\lambda$.
As a result, several alternative approaches have been proposed to perform
AUC scoring; e.g. one may replace it with a convex approximation \citep[Sect. 7]{ustats}
or use a PAC-Bayesian approach as in \citet{pac_auc}.

Regarding point (b), we use the inverse $\psi^{-1}$ of the stereographic projection, which
transforms a vector $x\in \R^{p-1}$ into a point $\theta = \psi^{-1}(x)$ on the
unit hyper-sphere in dimension $p$, $\mathbb{S}^p = \{\theta \in R^p,
	\|\theta\|=1\}$. That is, 
if $\theta=\psi^{-1}(x)$, then
$\theta_i = 2 x_i / \sum_{j=1}^d x_j^2$, for $i=1,\dots, p-1$,
$\theta_p = (\sum_j x_j^2 - 1 ) / (\sum_j x_j^2 + 1)$.
Then we define the objective function to be $l(x) = u\circ\psi^{-1}(x)$; thus $l:\R^d \to \R$ 
with $d=p-1$. It is reasonably straightforward to show that this function is strongly 
lower semi-continuous, as $\psi^{-1}$ is Lipschitz, and $u$ is a sum of indicator functions. 

We consider two classical datasets from the UCI machine learning repository
(\url{https://archive.ics.uci.edu/}): Pima (short for Pima Indians Diabetes),
and Sonar. We pre-process the data so that each predictor is normalised, i.e.
the empirical mean is set to zero and the variance is set to one. We compare
our algorithm to a strategy often used in practice which  relies on the
Nelder-Mead, or simplex, algorithm with random start. This approach is
considered naive in that Nelder-Mead does not require differentiability for
implementation, but is a requirement for correctness. The stepsizes are set to
$\gamma_{n}=0.2/(1+n)^{0.5}$. We run the two algorithms $10$ times.


In \cref{fig:Pima}, left panel, we  report the running estimates for
the Pima Indian dataset ($p=8$). One can see that the estimates converge very quickly 
for this dataset, despite the fact that $\gamma_n$ converges slowly.

\begin{figure}
	\centering{}
	\includegraphics[scale=0.28]{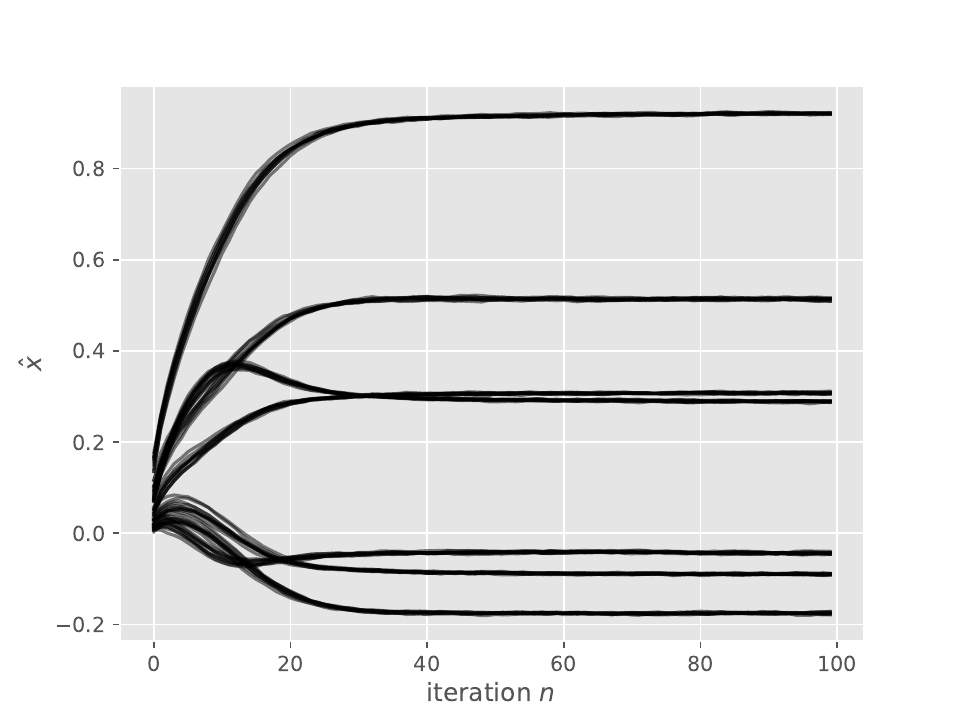}
	\includegraphics[scale=0.28]{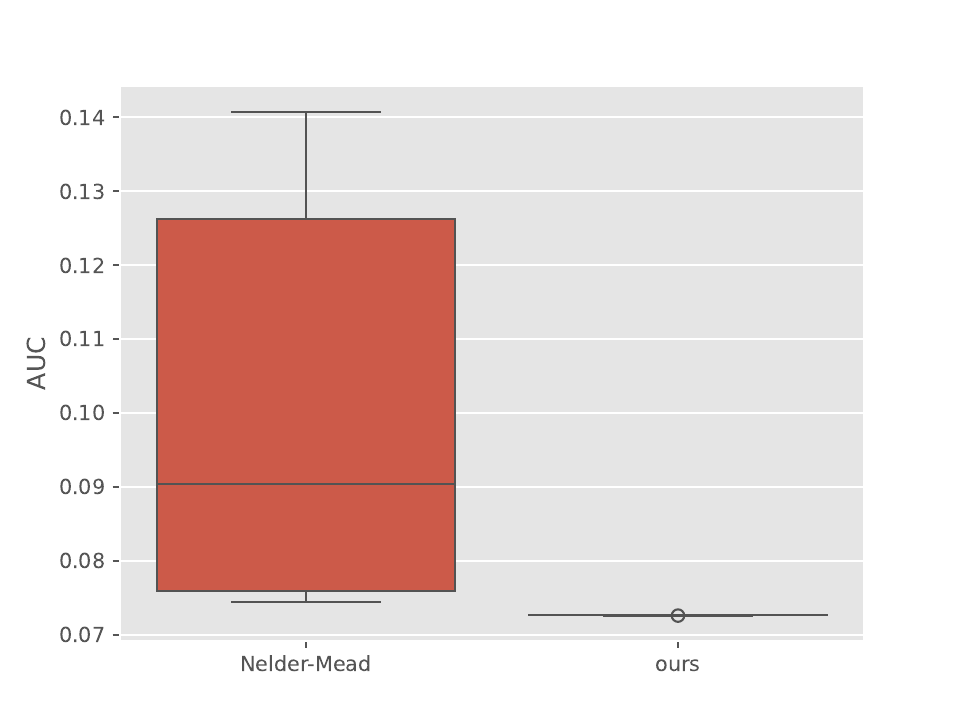}
	\includegraphics[scale=0.28]{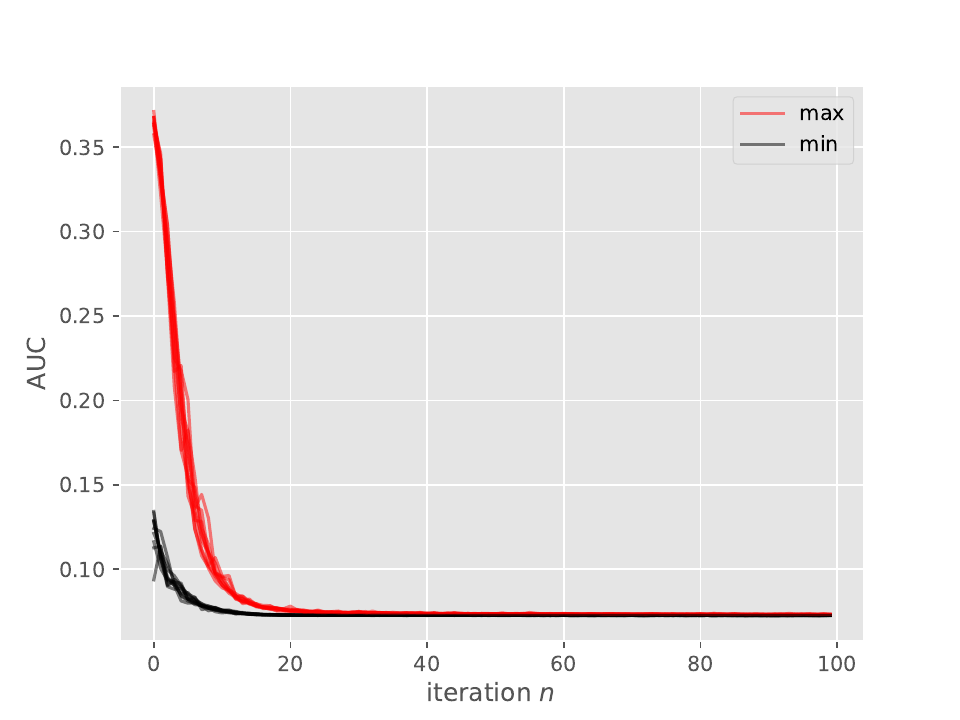}
	\caption{\label{fig:Pima}Pima dataset (10 runs). Left: running mean of the Gaussian
	distribution as a function of iteration $n$ (different lines correspond to different runs).
		Center: box-plots of the estimated minimum AUC for the two considered algorithms.
		Right: min/max of $l$ across iterations $n$ (again, different lines correspond to different runs).
	}
\end{figure}

In \cref{fig:Pima}, centre panel, we compare the variability (over $10$ runs)
of the output of the two considered algorithms.  One can see that our algorithm
provides much lower empirical risk than the naive approach based on
Nelder-Mead. Finally, the right panel reports the smallest and largest values
of $l(X_n^i)$ at iteration $n$ of our algorithm.

We repeat this experiment for the Sonar dataset ($p=60$); see \cref{fig:Sonar}
for the same plots as in \cref{fig:Pima} for this second dataset.  We notice
that a significantly larger number of iterations is required to achieve
convergence in this case (maximum number of iterations was set to $4\times
10^3$). The slower convergence may be due to several factors: first, $p=60$,
i.e. the dimension is higher than in the previous example. Second, the dataset
exhibits near complete separation: the best AUC score we observe is below
$1.5\times 10^{-3}$. Third, the function reaches its minimal value on a very
tiny region (of radius of order $10^{-3}$); notice in particular how the
highest AUC value remains high for a long time in the right panel of
\cref{fig:Sonar}.

\begin{figure}
	\centering{}
	\includegraphics[scale=0.28]{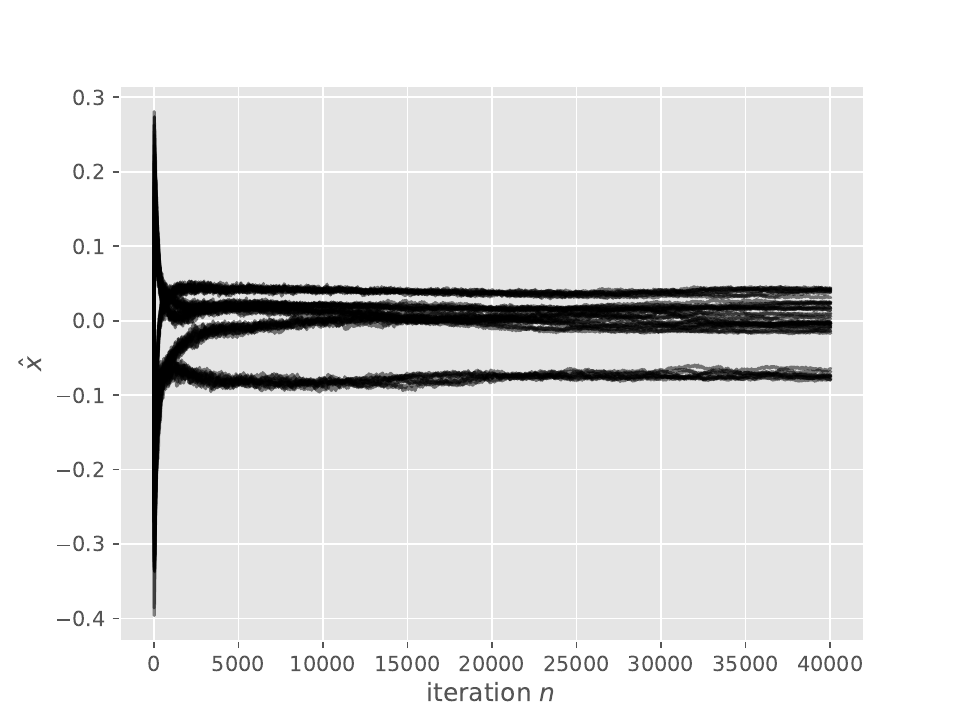}
	\includegraphics[scale=0.28]{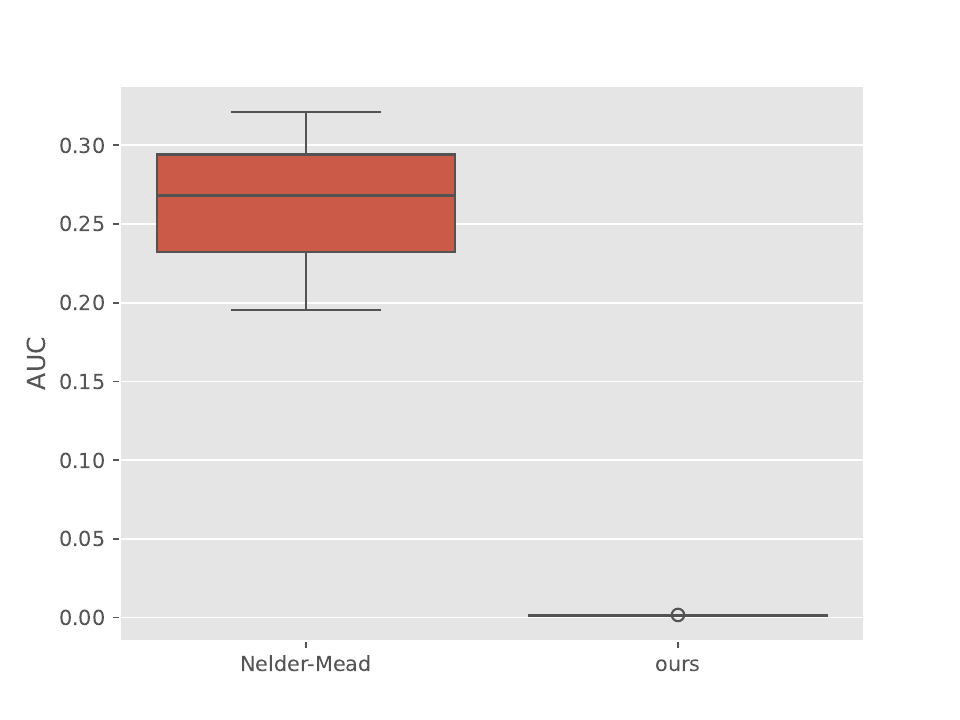}
	\includegraphics[scale=0.28]{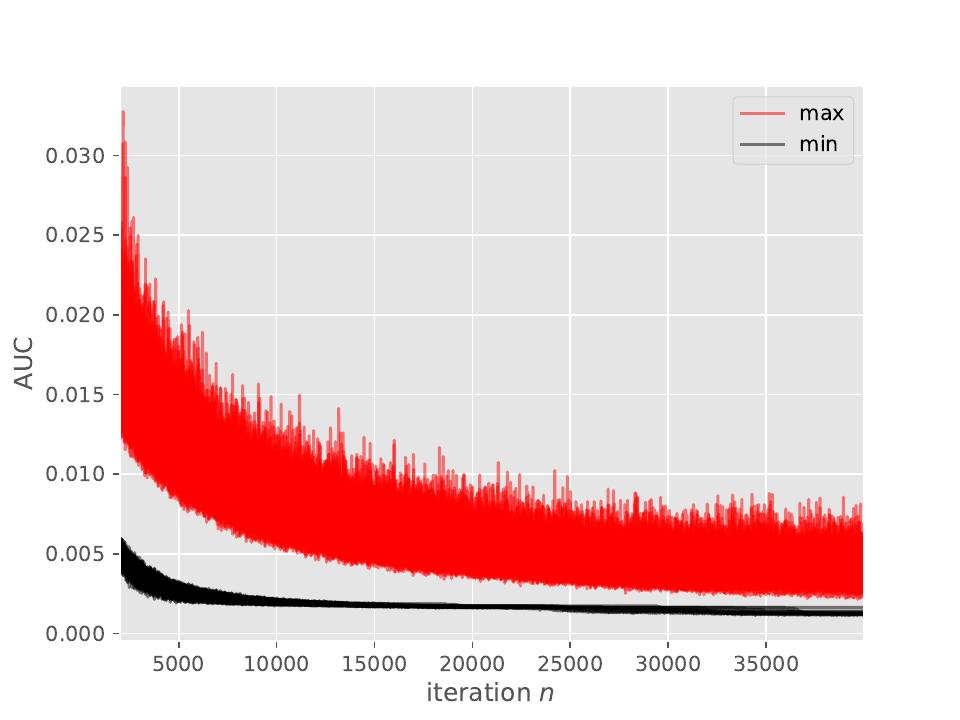}
	\caption{\label{fig:Sonar}
		Same plots as ~\cref{fig:Pima} for the Sonar dataset. For the left plot, only the last
		5 components are shown. 
		For the right plot, the first 400 iterations are not plotted to improve readibility.}
\end{figure}

\section{Discussion\label{discussion}}

We briefly discuss possible extensions and variations we are
currently exploring or have not explored yet.

\subsection{Accounting for noise}
\label{sub:accounting_noise}

We have established the convergence of the \emph{ideal}
\cref{alg:main_det_algo};  what remains to be done is doing the same for the
practical variant where  the distributions $\pi_n$ and $\tilde{\pi}_n$  are
replaced by Monte Carlo approximations.

Another way to account for noise is to consider the scenario where the objective function
can be evaluated only up to some noise. That is, one cannot evaluate $l$ exactly, but one can
evaluate a random estimate $\ell(x, U)$ such that $\E[\ell(x, U)]=l(x)$, where $U$
has a certain distribution. In this case, the ideal Alg. \ref{alg:main_det_algo} can be easily adapted by considering, at each iteration $n\geq 0$, the objective function $\ell(x,U_n)$, where $U_n$ has been sampled (once) at iteration $n$,
from the distribution of $U$. Accordingly, its practical implementations presented in Appendix \ref{app:additional-sim-results} can be adapted by
replacing the instruction for computing the weights with:    $w_n^i \gets
	\exp\{-\ell(X_n^i, U_n)\}$ (thus, the $N$ evaluations at iteration $n$ are
performed using the same random variable $U_n$).

We are currently developing theory for this type of noisy optimisation
scenario. In the stochastic framework the algorithm proceeds as in \cref{alg:main_det_algo}
and the sequence $\{\theta_{n},n\geq1\}$ is now the output of a
time inhomogeneous stochastic gradient algorithm
\[
	\theta_{n+1}=\theta_{n}-\gamma_{n}\nabla_{\theta}\left\{-\log\int e^{-\ell(x,U_{n})}\phi\biggl(\frac{x-\theta_{n}}{\sqrt{\gamma_{n}}}\biggr)\mathrm{d}x\right\},\ U_{n},\ n\geq1\overset{{\mathrm iid}}{\sim}P\,,
\]
here presented in the Gaussian scenario.




\subsection{A Proximal version}
\label{sub:proximal_version}

For completeness, we discuss here a variation on Algorithm \ref{alg:main_det_algo} which can be shown to be equivalent to a time-inhomogeneous proximal minimisation recursion. A natural modification of this algorithm consists of inverting the order of the distributions in the minimisation step of the $\mathrm{KL}$ divergence; see Algorithm~\ref{alg:proximal_modification}.

\begin{algorithm}[H]
	\caption{Proximal version}
	\label{alg:proximal_modification}
	\begin{algorithmic}
		\Require objective function $l$, initial parameter $\theta_0$, stepsizes
		$(\gamma_n)_{n\geq 0}$.
		\State $\pi_0 \gets \pi_{\theta_0, \gamma_0}$

		\While{$n \geq 0$}

		\State $\tilde{\pi}_{n+1}(x)\propto \exp\left\{-l(x)\right\}\pi_{n}(x)$

		\State $\theta_{n+1}\in \argmin_{\theta\in \Theta} \mathcal{\mathrm{KL}}(\pi_{\theta,\gamma_{n}}, \tilde{\pi}_{n+1})$
		\Comment{Swap distributions in KL (compared to \cref{alg:main_det_algo})}

		\State $\pi_{n+1} \gets \pi_{\theta_{n+1},\gamma_{n+1}}$

		\EndWhile

		\Ensure sequence of distributions $\tilde{\pi}_n$ and parameters $\theta_n$.
	\end{algorithmic}
\end{algorithm}
Link to proximal minimisation can be derived via an argument analogous to that in \cref{lemma: laplace principle}, which demonstrates that Algorithm~\ref{alg:proximal_modification} generates a sequence $\crl{\theta_n\in \R^d, \ n\in \mathbb{N}}$ such that
\begin{align*} \theta_{n+1}=\argmin_{\theta}\crl{F_{\gamma_n}(\theta)+\mathrm{KL}(\pi_{\theta, \gamma_n},\pi_{\theta_{n}, \gamma_n})}, \ n\in \mathbb{N},
\end{align*}
where $F_{\gamma}(\theta):=\int l(x)\pi_{\theta,\gamma}(x)\mathrm{d}x$, $\theta\in \R^d$, $\gamma>0$. We see that the algorithm aims to adjust $\pi_{\theta,\gamma}$ to decrease $\theta\mapsto F_\gamma(\theta)$ subject to a proximal penalty, which is reminiscent of proximal Expectation-Maximization (EM) algorithms. When $\pi_{\theta,\gamma}$ is a Gaussian,  $\mathrm{KL}(\pi_{\theta, \gamma_n},\pi_{\theta_{n}, \gamma_n})=(2\gamma_n)^{-1}||\theta-\theta_n||^2$, and interpretation in terms of an inhomogeneous (vanilla) proximal minimisation algorithm should be clear. Note that, however, in the general case of exponential families we have that for any $n\in \mathbb{N}$,
$
	\mathrm{KL}(\pi_{\theta, \gamma_n},\pi_{\theta_{n}, \gamma_n})
$ is the Bregman divergence $\mathrm{D}_A(\theta_{n},\theta)$ (see
\cref{lemma: kl and bregman}), convex in the first variable but not
necessarily in the second. Hence Algorithm~\ref{alg:proximal_modification} above cannot always be interpreted as a proximal algorithm.
Despite its theoretical attractiveness we have not pursued this proximal approach here since its implementation seems to require another minimisation procedure, in contrast with our approach. Finally we remark that Algorithm~\ref{alg:main_det_algo} can be interpreted as a coordinate descent algorithm to optimise the functional
$\Phi_1(\nu,\theta;\gamma)=\int l(x)\nu(x){\rm d}x+{\rm KL}(\nu,\pi_{\theta,\gamma})$, with $\nu$ a probability density, while Algorithm~\ref{alg:proximal_modification} corresponds to $\Phi_2(\nu,\theta;\gamma)=\int l(x)\pi_{\theta,\gamma}(x){\rm d}x+{\rm KL}(\pi_{\theta,\gamma},\nu)$, where we notice that optimising $\nu$ for $\pi_{\theta,\gamma}$ fixed leads to $\nu = \pi_{\theta,\gamma}$, therefore agreeing with the interpretation above.

\begin{appendix}

	\section{Notation} \label{sec:notation}
	List of notation:
	\begin{itemize}
		\item $\R^d$, $d\geq 1$, denotes the real coordinate $d$-space, and  $\mathcal{B}(\R^d)$ its Borel sigma-algebra.
		\item $\R_+$ denotes the set of positive real numbers including zero.
		\item $\emptyset$ denotes the empty set.
		\item $\mathbb{M}^1_+(\R^d)$ is the set of probability measures on the space $\R^d$.
		\item For $\nu, \rho \in \mathbb{M}^1_+(\R^d)$, we write $\nu \ll \rho$ if $\nu$ is absolutely continuous with respect to $\rho$.
		\item For a measure $\pi\in \mathbb{M}^1_+(\R^d)$ that is absolutely continuous with respect to the Lebesgue measure, we again denote by $\pi$ its density. In formulae, we write $\pi(\mathrm{d}x)=\pi(x)\mathrm{d}x$.
		\item We sometimes adopt the linear functional notation for integrals with respect to measures,  writing $\pi(f)$ for $\int f(x)\pi(\mathrm{d}x)$, for measurable functions $f$.
		\item We denote expectation operator by $\E$. When we need to specify the probability measure $\nu$ of integration, we sometimes write $\E_{\nu}f(X)$ for $\nu(f)$.
		\item For $\mu,\nu \in \mathbb{M}^1_+(\R^d)$ with absolutely continuous density (with respect to the Lebesgue measure), we denote the Kullback-Leibler divergence by $\mathrm{KL}\round{\mu,\nu}:=\int \mu(x)\log\frac{\mu(x)}{\nu(x)}\mathrm{d}x$.


		\item Given a point $\theta\in \R^d$ and a sequence $\crl{\theta_n\in \R^d, n\in \mathbb{N}}$, we say that $\underset{n\to \infty}{\lim} \theta_n=\theta$ if for every $\epsilon>0$ there exists number $n_0=n_0(\epsilon)\in \mathbb{N}$ such that, for every $n \geq n_0$, $||\theta_n-\theta||< \epsilon.$ We also write, with the same meaning, $\theta_n\to_n \theta$ or $\theta_n\to \theta$ when there is no ambiguity. As a shortcut, we sometimes denote the limit operation $\underset{n\to \infty}{\lim} $ as $\underset{n}{\lim}$.
		\item Similarly, we write $\underset{n}{\liminf}, \ \underset{n}{\limsup}$ for $\underset{n\to \infty}{\liminf}, \  \underset{n\to \infty}{\limsup}$, respectively.
		\item For a sequence $\crl{\alpha_n\in \R_+, \ n\in \mathbb{N}}$ that decreases to zero, we write $\alpha_n \downarrow 0$.
		\item The closure of a set $A\subset \R^d$ is denoted by $\mathrm{cl}(A)$ and corresponds to the intersection of all closed subsets of $\R^d$ containing $A$.
		\item For two sets $C_1,C_2$ of $\R^d$, we denote their elementwise (Minkowsky) sum by $C_1+C_2:=\crl{x+y; x\in C_1, y\in C_2}$
		\item Let $(\R^d)^{\mathbb{N}}$ denote the space of sequences $\crl{\theta_n\in \R^d, n\in \mathbb{N}}$. Let $\crl{\theta_n}\in (\R^d)^{\mathbb{N}}$. Given $g:(\R^d)^{\mathbb{N}}\to \R$, we write $\underset{\theta_n\to \theta}{\inf}g(\crl{\theta_n})$ for $\inf\crl{g(\crl{\theta_n}); \ \theta_n\in \R^d, \ n\in \mathbb{N}, \ \theta_n\to \theta}$.
		\item $C^1(\R^d)$ denotes the class of functions $f:\R^d\to \R$ with continuous gradient. Gradient operators are denoted by $\nabla$, or $\nabla_{\theta}$ when we need to specify that the variable of differentiation is $\theta$
		\item The $d$-dimensional standard normal density is denoted as $\phi(z):=\frac{1}{(2\pi)^{d/2}}e^{-\frac{||z||^2}{2}}$, $z\in \R^d$.
		\item Given two functions $f,g:\R^d\to \R$, we write $f\geq g$ if $f(\theta)\geq g(\theta)$ $\forall \theta \in \R^d$.
		\item The Dirac Delta measure on zero is denoted as $\delta_0(\mathrm{d}x)$.
		\item For two vectors $\theta,\theta'\in \R^d$, we denote their Euclidean inner product by $\inner{\theta}{\theta'}$.

		\item The notation $$\argmin f$$ stands for $\underset{u\in \R^d}{\argmin} f(u)$ and denotes the set of global minima of a function $f:\R^d\to \R$.
		\item The set of local minima of a function $f$ is denoted by $$\mathrm{loc-}\argmin f.$$
	\end{itemize}

	\section{Exponential family background}

	\subsection{Natural exponential families}
	\label{app: exponential families}

	For $x, \theta\in \R^d$, $T:\R^d\to\R^d$ and a baseline density $\upsilon$
	on $\R^d$ (with respect to some dominating measure, e.g., Lebesgue), we
	consider a (natural) exponential family to be a family of probability
	densities
	\[\pi_{\theta}(x)
		=\exp\crl{\inner{\theta}{T(x)}-A(\theta)}\upsilon(x)\,,
	\]
	where the cumulant (log-partition) function for $T(X)$ is
	\begin{equation}\label{eq:log_partition}
		A(\theta)=\log\int \exp\left\{\langle\theta,T(x)\rangle\right\}
		\upsilon(x)\mathrm{d}x\,,
	\end{equation}
	and the canonical parameter $\theta$ of interest belongs to the set $$\Theta:=\crl{\theta\in \R^d \colon \ A(\theta)<\infty}.$$

	\begin{definition}
		\label{def: minimal}
		Given an exponential family with sufficient statistic $T:\R^d\to \R^d$, we
		say that the family is minimal if the elements of $T$ are linearly
		independent, that is, if there is no nonzero vector $a\in \R^d$ s.t.
		$\sum_{i=1}^d a_{i}T_{i}(x)$ is equal to a constant almost
		everywhere. This implies that there is a unique natural parameter vector
		$\theta$ associated with each distribution.
	\end{definition}

	\begin{definition}
		\label{def: regular}
		An exponential family with log-partition function $\theta\mapsto A(\theta)$ is said to be regular when the domain $\Theta$ is an open set.
	\end{definition}

	Examples of minimal and regular exponential families include Bernoulli, Gaussian, Exponential, Poisson, and Beta distributions.


	\begin{prop}{\cite[Proposition 3.1]{wainwrightjordan2008}.}
		The log-partition function \eqref{eq:log_partition}
		associated with any regular exponential family  with sufficient statistic
		$T:\R^d\to \R^d$ has the following properties:
		\begin{enumerate}
			\item It has derivatives of all orders on its domain $\Theta$. Furthermore,
			      \begin{align*}
				      \nabla A(\theta)   & = \E_{\pi_{\theta}}\left[T(X)\right],           \\
				      \nabla^2 A(\theta) & = \mathrm{var}_{\pi_{\theta}}\left[T(X)\right].
			      \end{align*}
			\item $\theta\mapsto A(\theta)$ is a convex function on $\Theta$, and strictly convex if the representation is minimal.
		\end{enumerate}
	\end{prop}

	The convexity argument comes from the fact that the full Hessian $\nabla^2 A(\theta)$ is the covariance matrix of the random vector $T(X)$, and so is positive semidefinite on the open set $\Theta$, which ensures convexity.\\

	We now report an important dual coupling property of exponential families. Let
	\[\mathcal{M}:=\crl{\mu\in \R^d; \ \exists \ \theta  \ \text{s.t.} \
			\E_{\pi_{\theta}}T(X)=\mu}
	\]
	be the set of so-called moment parameters. We have
	\begin{prop}{\cite[Proposition 3.2]{wainwrightjordan2008}}
		The gradient mapping $\nabla A:\Theta\to \mathcal{M}$ is one-to-one if and only if the exponential representation is minimal.
	\end{prop}

	\begin{thm}{\cite[Theorem 3.3]{wainwrightjordan2008}}
		In a minimal exponential family, the gradient map $\nabla A$ is onto the interior of $\mathcal{M}$, denoted by $\mathcal{M}^o$. Consequently, for each $\mu\in \mathcal{M}^o$, there exists some $\theta=\theta(\mu)\in \Theta$ such that $\E_{\pi_{\theta}}T(X)=\mu$.
	\end{thm}




	We conclude with the following relation for the $\mathrm{KL}$ and Bregman divergence in the case of exponential families. A more complete statement which also includes the so-called dual function of $A$ and the dual parameters is also available. See for instance \cite[Section 4]{nielsen_nock_2010} for more details.
	\begin{lemma}{\cite{nielsen_nock_2010}.}
		\label{lemma: kl and bregman}
		Assume $\pi_{\theta_1}$ and $\pi_{\theta_2}$ belong the same minimal regular exponential family with log-partition function $A$. Suppose $\theta_1$, $\theta_2$ are their natural parameters. We have
		\begin{equation}
			\mathrm{D}_{A}(\theta_2,\theta_1)=\mathrm{KL}(\pi_{\theta_1}, \pi_{\theta_2}).
		\end{equation}
	\end{lemma}

	\subsection{Exponential dispersion models}
	\label{sub:edm_details}

	Consider an exponential family, with base measure $\upsilon$, canonical
	statistic $T:\R^d\rightarrow \R$,
	and cumulant (log-partition) function $A$. We may extend this family by
	considering an EDM (exponential dispersion model) family
	as follows \citep{Jorgensen:1987}.
	Consider values $\gamma\in\Gamma\subset\mathbb{R}_{+}$ such that
	$\gamma^{-1} A(\theta)$ is the cumulant function of some ($\gamma$-dependent)
	probability density $\upsilon_{\gamma}$, and define:
	\[
		\pi_{\theta,\gamma}^{\star}(x) \coloneq
		\exp\left\{\langle\theta,T(x)\rangle-\gamma^{-1} A(\theta)\right\}
		\upsilon_{\gamma}(x)\,.
	\]
	If we assume $T(x)=x$, and apply the change of variable $x'= \gamma \times x$,
	we obtain the following EDM distribution associated with $(\theta, \gamma)$:
	\[
		\pi_{\theta,\gamma}(x') \coloneq
		\exp\left\{\frac{1}{\gamma}[\langle\theta,x'\rangle - A(\theta)]\right\}
		\upsilon'_{\gamma}(x')\,,
		\quad \upsilon'_\gamma\coloneq \gamma^{-d} \upsilon_\gamma \,.
	\]
	This family is of interest to us because for $\theta\in\Theta$ all
	moments of this distribution exist:
	\begin{gather*}
		\mu(\theta)  =\mathbb{E}_{\pi_{\theta,\gamma}}\left(X\right)=\nabla_{\theta}A(\theta)\\
		\var\big(X\big)  =\gamma\nabla_{\theta}^{2}A(\theta)
	\end{gather*}
	therefore implying concentration of $\pi_{\theta,\gamma}$
	on $\mu(\theta)$ as $\gamma\downarrow0$.
	It is then possible to use the corresponding EDM within
	\cref{alg:main_det_algo}, as explained in \cref{sub:beyond_gauss}.

	The general case ($T(x)\neq x$) may be worked out along the same lines. In
	practice, however, it is often simpler to re-express the model as a family of
	probability distribution for variable $x'=T(x)$ first. Next section works out
	how to derive an EDM family from Wishart distributions.

	\subsection{EDM families and Wishart distributions}
	\label{sub:wishart-example}
    Let $\Sigma$ be a $d\times d$ symmetric positive matrix, and $\nu > d-1$. A
Wishart distribution with parameter $(\Sigma, \nu)$ is a distribution defined
over the set of positive definite matrices $Y$ of size $d\times d$, with density:
\begin{equation}
  \label{eq:wishart_density}
  \pi_\theta(Y) =
 \frac{\exp\left\{-\frac{1}{2}{\mathrm Tr}\big(\Sigma^{-1}Y\big)
   - \frac{\nu}{2} \log |\Sigma| \right\}
 |Y|^{\frac{\nu-d-1}{2}}}{2^{\frac{\nu
d}{2}}\Gamma_{d}(\nu/2)}
\end{equation}
where the supporting measure is the Lebesgue measure $\mathrm{\upsilon}$ on
$\R^{d(d+1)/2}$, to account for the symmetry of $Y$.
\begin{prop}
\label{prop:rescale_wishart_to_edm}
    
Let $Y$ be distributed as in (\ref{eq:wishart_density}). Let $0<\gamma=\nu^{-1}$ and apply the change of variable
$X=\gamma Y$. Then, with $\theta:=-\frac{1}{2}\Sigma^{-1}$ and $A(\theta):=-\frac{1}{2}\log|-\theta|$ , one obtains that $X$ has density with respect to the Lebesgue measure $\mathrm{\upsilon}$ given by
\begin{equation}
\label{eq:rescaled_wishart}
 \pi_{\theta,\gamma}(X)=\frac{\exp\crl{\gamma^{-1}\sqrd{\inner{\theta}{X}_F-A(\theta)}}|X|^{\frac{\gamma^{-1}-d-1}{2}}}{\gamma^{\frac{d}{2\gamma}}\Gamma_{d}(1/2\gamma)}.
\end{equation}
Moreover, $\nabla_{\theta} A(\theta)=\Sigma$.
\end{prop}

\begin{proof}
    
Applying the change of variable to (\ref{eq:wishart_density}) one clearly obtains the density
\begin{align*}
  \pi_{\theta, \gamma}(X) &= 
 \frac{\exp\left\{-\gamma^{-1}\frac{1}{2}\mathrm {Tr}\big(\Sigma^{-1}X\big)
   - \gamma^{-1}\frac{1}{2} \log |\Sigma| \right\}
 \gamma^{-\frac{d(\nu-d-1)}{2}}|X|^{\frac{\nu-d-1}{2}}}{2^{\frac{\nu d}{2}}\Gamma_{d}(\nu/2)}\times \nu^{d(d+1)/2}\\
 &= \frac{\exp\left\{-\gamma^{-1}\frac{1}{2}\mathrm {Tr}\big(\Sigma^{-1}X\big)
   - \gamma^{-1}\frac{1}{2} \log |\Sigma| \right\}
 \gamma^{-\frac{d(\nu-d-1)}{2}}|X|^{\frac{\nu-d-1}{2}}}{2^{\frac{\nu d}{2}}\Gamma_{d}(\nu/2)}\times \gamma^{-d(d+1)/2}\\
 &= \frac{\exp\crl{-\gamma^{-1}\frac{1}{2}\inner{\Sigma^{-1}}{X}_F-\gamma^{-1}\frac{1}{2}\log|\Sigma|}|X|^{\frac{\nu-d-1}{2}}}{\gamma^{\frac{\nu d}{2}} 2^{\frac{\nu d}{2}}\Gamma_{d}(\nu/2)}\\
 &=\frac{\exp\crl{\gamma^{-1}\inner{-\frac{1}{2}\Sigma^{-1}}{X}_F-\gamma^{-1}\frac{1}{2}\log|\Sigma|}|X|^{\frac{\nu-d-1}{2}}}{(2\gamma)^{\frac{\nu d}{2}}\Gamma_{d}(\nu/2)}\\
 &=\frac{\exp\crl{\gamma^{-1}\sqrd{\inner{-\frac{1}{2}\Sigma^{-1}}{X}_F-\frac{1}{2}\log|\Sigma|}}|X|^{\frac{\nu-d-1}{2}}}{(2\gamma)^{\frac{\nu d}{2}}\Gamma_{d}(\nu/2)}.
\end{align*}
where the term $\nu^{d(d+1)/2}$ on the first row is the Jacobian coming from rescaling of the Lebesgue measure after the change of variable. Now note that
\begin{align*}
   -\log|\Sigma| = \log|\Sigma^{-1}| = \log(2^d)+\log|\frac{1}{2}\Sigma^{-1}|
\end{align*}

and $$e^{-\frac{1}{2\gamma}\log(2^d)}=\round{\frac{1}{2^d}}^{\frac{1}{2\gamma}}=\frac{1}{2^{\frac{d}{2\gamma}}}.$$
Using the above identities one can write
\begin{align*}
\frac{\exp\crl{\gamma^{-1}\sqrd{\inner{-\frac{1}{2}\Sigma^{-1}}{X}_F-\frac{1}{2}\log|\Sigma|}}|X|^{\frac{\nu-d-1}{2}}}{(2\gamma)^{\frac{\nu d}{2}}\Gamma_{d}(\nu/2)}&=\frac{\exp\crl{\gamma^{-1}\sqrd{\inner{-\frac{1}{2}\Sigma^{-1}}{X}_F-\round{-\frac{1}{2}\log|\frac{1}{2}\Sigma^{-1}|}}}|X|^{\frac{\nu-d-1}{2}}}{\frac{1}{2^{d/(2\gamma)}}(2\gamma)^{\frac{\nu d}{2}}\Gamma_{d}(\nu/2)}\\
&=\frac{\exp\crl{\gamma^{-1}\sqrd{\inner{-\frac{1}{2}\Sigma^{-1}}{X}_F-\round{-\frac{1}{2}\log|\frac{1}{2}\Sigma^{-1}|}}}|X|^{\frac{\nu-d-1}{2}}}{\gamma^{\frac{\nu d}{2}}\Gamma_{d}(\nu/2)}
\end{align*}
where the last equality holds as $\nu=\gamma^{-1}$.

Letting $\theta:=-\frac{1}{2}\Sigma^{-1}$, we see that for $A(\theta):=-\frac{1}{2}\log|-\theta|$ one can also write
\begin{align*}
    \frac{\exp\crl{\gamma^{-1}\sqrd{\inner{-\frac{1}{2}\Sigma^{-1}}{X}_F-\round{-\frac{1}{2}\log|\frac{1}{2}\Sigma^{-1}|}}}|X|^{\frac{\nu-d-1}{2}}}{\gamma^{\frac{\nu d}{2}}\Gamma_{d}(\nu/2)}=\frac{\exp\crl{\gamma^{-1}\sqrd{\inner{\theta}{X}_F-A(\theta)}}|X|^{\frac{\nu-d-1}{2}}}{\gamma^{\frac{\nu d}{2}}\Gamma_{d}(\nu/2)}
\end{align*}
This shows that under the rescaling we have
\begin{align*}
    \pi_{\theta,\gamma}(X)=\frac{\exp\crl{\gamma^{-1}\sqrd{\inner{\theta}{X}_F-A(\theta)}}|X|^{\frac{\nu-d-1}{2}}}{\gamma^{\frac{\nu d}{2}}\Gamma_{d}(\nu/2)}=\frac{\exp\crl{\gamma^{-1}\sqrd{\inner{\theta}{X}_F-A(\theta)}}|X|^{\frac{\gamma^{-1}-d-1}{2}}}{\gamma^{\frac{d}{2\gamma}}\Gamma_{d}(1/2\gamma)}.
\end{align*}

For the last statement, use the property that for any positive symmetric matrix $B$ $$\nabla_{B}\log|B|=(B^{-1})^T$$ to note that (using $\theta=-\frac{1}{2}\Sigma^{-1}$ and  symmetry)
\begin{align*}
    \nabla_{\theta} A(\theta)=-\frac{1}{2}\nabla_{\theta}\log|-\theta| 
     =-\frac{1}{2}\nabla_{\theta}\log|\frac{1}{2}\Sigma^{-1}|
    =\frac{1}{2}\nabla_{-\theta}\log |\frac{1}{2}\Sigma^{-1}|
    &=\frac{1}{2}\nabla_{\frac{1}{2}\Sigma^{-1}}\log |\frac{1}{2}\Sigma^{-1}|\\
    &=\frac{1}{2}\sqrd{\round{\frac{1}{2}\Sigma^{-1}}^{-1}}^T\\
    &=\Sigma.
\end{align*}

\end{proof}

Note that $\pi_{\theta,\gamma}(X)$ can be viewed as an EDM family, with parameter $(\theta, \gamma)$.
 Using standard properties of Wishart distributions,  we have
$\E[X]=\Sigma$, $\operatorname{var}[X_{i, j}] = \gamma (\Sigma_{ij}^2 +
\Sigma_{ii}\Sigma{jj})$, and concentration occurs as $\gamma\to 0$.
The important point here is that it is the family of \emph{rescaled}
Wishart distributions (i.e., distributions for $X=\gamma Y=Y/\nu$) 
which may be viewed as an EDM, not the original family.
In this
case, not all values of $\gamma$ are permitted; i.e. since $\gamma=1/\nu$ and
$\nu > d -1$, one has $\gamma < 1/(d-1)$, but that does not pose any practical
problem.

	\section{Lower-semicontinuity and epigraphs}
	\label{app:lsc_epigraphs}

	Lower semi-continuity can be also defined as a property of certain sets. We report this characterisation here.

	\begin{definition}
		Let $f:\R^d\to \R$ be a function. We denote
		\begin{enumerate}
			\item The epigraph set of $f$ by
			      \begin{equation*}
				      \mathrm{epi}(f):=\crl{(\theta,a)\in \R^d\times \R; \ f(\theta)\leq a}.
			      \end{equation*}
			\item
			      The hypograph set of $f$ by
			      \begin{equation*}
				      \mathrm{hypo}(f):=\crl{(\theta,a)\in \R^d\times \R; \ f(\theta)\geq a }.
			      \end{equation*}
			\item The level sets of $f$ by
			      $$\mathrm{lev}_{\leq \alpha}(f):=\crl{\theta\in \R^d; \ f(\theta)\leq \alpha}, \ \  \alpha\in \R.$$
		\end{enumerate}
	\end{definition}

	The following theorem characterises lower semi-continuity in terms of epigraphs and level-sets. A similar result can be stated in terms of upper semi-continuity and hypographs.
	\begin{thm}{\cite[Theorem 1.6]{RockWets98}}
		For a function $f:\R^d\to \R$, the following statements are equivalent
		\begin{itemize}
			\item $f$ is lower semi-continuous on $\R^d$
			\item Its epigraph set $\mathrm{epi}(f)$
			      is closed in $\R^d\times \R$
			\item Level sets $\mathrm{lev}_{\leq a}$
			      are closed in $\R^d$, for each $a\in \R$.
		\end{itemize}
	\end{thm}
	We report two additional examples which constitute a relevant theoretical tool in this manuscript. More details on the quantities below can be found in \citet[Chapter 1, Section D]{RockWets98}.
	\begin{definition}
		\label{def: epi_closure}
		Consider a function $f:\R^d\to \R$. We define the epigraphical closure of $f$ as the function $$\theta\mapsto \mathrm{cl_e} \ f (\theta):= \inf_{\theta_k\to \theta}\liminf_k f(\theta_k)$$
		and the hypographical closure $f$ as
		$$\theta\mapsto \mathrm{cl_h} f(\theta):=\sup_{\theta_k\to \theta}\limsup_k f(\theta_k).$$
	\end{definition}

	\begin{prop}
		\label{prop: epi_closure}
		Consider a function $f:\R^d\to \R$. Then,

		\begin{itemize}
			\item $\mathrm{cl_e} f$ is lower semi-continuous with $\mathrm{cl_e} f(\theta)\leq f(\theta)$
			\item$\mathrm{cl_h} f$ is upper semi-continuous and $\mathrm{cl_h} f(\theta)\geq f(\theta)$.
		\end{itemize}
		Moreover,
		\begin{itemize}
			\item When $f$ is lower semi-continuous, $f(\theta) = \mathrm{cl_e}f(\theta)$
			\item When $f$ is upper semi-continuous, $f(\theta) = \mathrm{cl_h}f(\theta)$.
		\end{itemize}
	\end{prop}

	\section{Laplace Principle}
	\label{sub: proof of Laplace principle}

	The Laplace Principle is a known result which provides a variational representation of integrals of the form $-\log\int e^{-l(x)}\pi(x)\mathrm{d}x$, where $\pi$ is a probability density and $l$ is an integrable function.
	\begin{lemma}[Laplace Principle]
		\label{lemma: laplace principle}
		Let $\nu, \pi$ be two probability measures with $\nu \ll \pi$. Let $l: \R^d\to \R$ be a locally integrable function such that $\int e^{-l(x)}\pi(x)\mathrm{d}x<\infty$ and set $\Tilde{\pi}\propto e^{-l}\pi$ to be a probability measure. It holds that
		\begin{equation}
			\label{eq: laplace principle}
			\mathrm{KL}(\nu, \Tilde{\pi})=\int l(x)\nu(x)\mathrm{d}x+\mathrm{KL}(\nu,\pi)+\log\int e^{-l(x)}\pi(x)\mathrm{d}x.
		\end{equation}
		A well-known consequence is
		\begin{equation}
			\label{eq: bayes}
			\Tilde{\pi}=\argmin_{\nu \in \mathbb{M}^1_+(\R^d)}\crl{\int l(x)\nu(x)\mathrm{d}x+\mathrm{KL}(\nu,\pi)}.
		\end{equation}
	\end{lemma}

	\begin{proof}

		By direct calculations, we have
		\begin{align*}
			\mathrm{KL}(\nu,\Tilde{\pi}) & =\int \nu(x)\log\round{\frac{\nu(x)\pi(e^{-l})}{e^{-l(x)}\pi(x)}}\mathrm{d}x=                        \\
			                             & =\int \nu(x)\log\round{\frac{\nu(x)}{e^{-l(x)}\pi(x)}}\mathrm{d}x+\log \pi(e^{-l})                   \\
			                             & =\int \nu(x)\crl{\log(\nu(x))-\log(e^{-l(x)})-\log(\pi(x))}\mathrm{d}x+\log \pi(e^{-l})              \\
			                             & = \int \nu(x)\log\round{\frac{\nu(x)}{\pi(x)}}\mathrm{d}x+\int \nu(x)l(x)\mathrm{d}x+\log\pi(e^{-l})
		\end{align*}
		which leads to equation \eqref{eq: laplace principle}. This implies
		\begin{align*}
			\argmin_{\nu \in \mathbb{M}^1_+(\R^d)}\crl{\int l(x)\nu(x)\mathrm{d}x+\mathrm{KL}(\nu,\pi)}=\argmin_{\nu \in \mathbb{M}^1_+(\R^d)} \mathrm{KL}\round{\nu,\Tilde{\pi}}=\Tilde{\pi}
		\end{align*}
		as $\mathrm{KL}(\nu,\Tilde{\pi})=0$ if and only if $\nu=\Tilde{\pi}$. Hence
		relation \eqref{eq: bayes} holds.


	\end{proof}

	\section{Practical algorithms} \label{app:additional-sim-results}

	\Cref{alg:MC_version_det_algo} describes one of the practical algorithms that
	one may use to approximate the ideal \cref{alg:main_det_algo} in the Gaussian
	case. This algorithm relies on basic Monte Carlo. To use instead randomised
	quasi-Monte Carlo (as we did in our numerical experiments) in order to reduce
	the variability of the output, one may use instead \cref{alg:QMC_version}.
	For an overview of RQMC, see, e.g., the book of \cite{Lemieux_qmc_book}.

	\begin{algorithm}[t]
		\caption{Gradient-Free Algorithm, Monte Carlo variant (Gaussian case)}
		\label{alg:MC_version_det_algo}
		\begin{algorithmic}
			\Require objective function $l$, initial parameter $\theta_0$, stepsizes
			$(\gamma_k)_{k\geq 0}$, Monte Carlo sample size $N$.
			\While{$n \geq 0$}
			\State Sample  $X_n^i\sim N\left(\theta_n, \gamma_n I_d\right)$ for
			$i=1,\dots,N$
			\Comment{Monte Carlo approximation of $\pi_n$}
			\State  $w_{n}^{i} \gets \exp\left\{- l(X_n^i)\right\}$ for $i=1,\dots,N$
			\Comment{weighted sample $(X_n^i, w_{n}^i)_{i=1,\dots,N}$ approximates
				$\tilde{\pi}_{n+1}$}
			\State $\theta_{n+1} \gets \sum_{i=1}^N w_{n}^i X_n^i / \sum_{j=1}^N w_{n}^j$
			\Comment{Maximum likelihood estimation for Gaussian $N(\theta, \gamma_n I)$}
			\EndWhile
			\Ensure sequence of distributions $\tilde{\pi}_n$ and parameters $\theta_n$.
		\end{algorithmic}

	\end{algorithm}

	\begin{algorithm}[t]
		\caption{Gradient-Free Algorithm, RQMC (randomised quasi-Monte Carlo) variant (Gaussian case)}
		\label{alg:QMC_version}

		\begin{algorithmic}
			\Require objective function $l$, initial parameter $\theta_0$, stepsizes
			$(\gamma_k)_{k\geq 0}$, Monte Carlo sample size $N=2^k$.
			\While{$n \geq 0$}
			\State Generate RQMC (e.g., scrambled Sobol) point set $(V_n^1, \dots,
				V_n^N)$ of length $N$, dimension $d$.
			\State $X_n^i\gets \theta_n + \Phi^{-1}(V_n^i)$ for $i=1,\dots,N$
			\Comment{$\Phi^{-1}(v)=\left(\Phi^{-1}(v_1),\dots, \Phi^{-1}(v_d)\right)$}
			\State  $w_{n}^{i} \gets \exp\left\{- l(X_n^i)\right\}$ for $i=1,\dots,N$
			\Comment{weighted sample $(X_n^i, w_{n}^i)_{i=1,\dots,N}$ approximates
				$\tilde{\pi}_{n+1}$}
			\State $\theta_{n+1} \gets \sum_{i=1}^N w_{n}^i X_n^i / \sum_{j=1}^N w_{n}^j$
			\Comment{Maximum likelihood estimate for model $N(\theta, \gamma_n I)$}
			\EndWhile
			\Ensure sequence of distributions $\tilde{\pi}_n$ and parameters $\theta_n$.
		\end{algorithmic}

	\end{algorithm}


	To extend \cref{alg:MC_version_det_algo,alg:QMC_version} to other
	(non-Gaussian) EDM families (see \cref{sub:edm_details}), one may adapt the
	expression  for the estimate $\theta_{n+1}$ as follows:
	
	\begin{equation*}
		\theta_{n+1} \gets (\nabla A)^{-1}\left( \frac{\sum_{i=1}^N w_n^i
				T(X_n^i)}{\sum_{i=1}^N w_n^i}\right).
	\end{equation*}

\end{appendix}

\bibliography{complete}       


\end{document}